\documentclass{article}
\usepackage[final]{neurips_2024}

\usepackage[utf8]{inputenc}
\usepackage{microtype}
\usepackage{amsmath,amssymb,amsfonts,amsthm}
\usepackage{booktabs}
\usepackage{array}
\usepackage{graphicx}
\usepackage{float}
\usepackage{subcaption}
\usepackage{enumitem}
\usepackage{xcolor}
\usepackage{url}
\usepackage[colorlinks=true,citecolor=blue,linkcolor=blue,urlcolor=blue]{hyperref}
\usepackage[capitalize,noabbrev]{cleveref}
\usepackage[numbers,sort&compress]{natbib}

\setlist{leftmargin=1.25em,itemsep=0pt,topsep=1pt,parsep=0pt,partopsep=0pt}
\makeatletter
\renewcommand{\section}{\@startsection{section}{1}{0pt}%
  {0.6ex plus 0.2ex minus 0.2ex}%
  {0.35ex plus 0.15ex}%
  {\normalfont\Large\bfseries}}
\renewcommand{\subsection}{\@startsection{subsection}{2}{0pt}%
  {0.45ex plus 0.15ex minus 0.15ex}%
  {0.25ex plus 0.1ex}%
  {\normalfont\large\bfseries}}
\makeatother

\theoremstyle{definition}

\newtheorem{theorem}{Theorem}
\newtheorem{proposition}{Proposition}
\newtheorem{corollary}{Corollary}
\newtheorem{assumption}{Assumption}
\theoremstyle{remark}
\newtheorem{remark}{Remark}

\newcommand{\D}{\mathcal{D}}
\newcommand{\Dobs}{\D_{\mathrm{obs}}}
\newcommand{\Mzero}{\ensuremath{\mathrm{M0}}}
\newcommand{\Mone}{\ensuremath{\mathrm{M1}}}
\newcommand{\Mtwo}{\ensuremath{\mathrm{M2}}}
\newcommand{\Mthree}{\ensuremath{\mathrm{M3}}}
\newcommand{\E}{\mathbb{E}}
\newcommand{\Prb}{\mathbb{P}}

\newcommand{\ind}{\mathbb{1}}

\newcommand{\probup}{p_{\uparrow}}
\newcommand{\tv}{\mathrm{TV}}
\newcommand{\dop}{\mathrm{do}}

\title{Martingale Doppelg\"anger-Eval: An Identification Framework for Auditing Candlestick Understanding in Vision-Language Models}

\author{%
Ziyao Wang\\
Department of Mathematics and Statistics\\
Texas Tech University\\
\texttt{ziywang@ttu.edu}
}

\begin{document}
\maketitle

\begin{abstract}
We introduce Martingale Doppelg\"anger-Eval, a public shadow-market benchmark for auditing whether vision-language models (VLMs) use candlestick evidence rather than extrapolate past trends. The central difficulty is identification: on real market histories, chart evidence and trend are strongly coupled, so an observational score cannot determine whether a fluent technical-analysis narrative is grounded in local visual evidence. We prove this limitation formally: no evaluation functional computed from observational chart--label data can distinguish a grounded responder from a trend-shortcut responder under strong coupling, whereas matched evidence interventions separate the same responders at an exponential rate and trend--label swaps provide an independent shortcut stress test. The benchmark therefore evaluates frozen VLMs on rendered OHLCV charts under four controlled mechanisms: a martingale-null market, injected-alpha counterfactual pairs, trend-confounder swaps, and regime shifts. A structural behavioral model identifies null-market bias, trend sensitivity, evidence sensitivity, prompt/renderer fragility, and evidence faithfulness; the accompanying statistical toolkit provides minimum detectable effects, block-aware sequential testing for metered APIs, and an overlap-weighted artifact check. Across frozen commercial and open VLMs, the identified regression assigns large positive coefficients to past trend but evidence coefficients that are zero or opposite to the rule-implied sign. Matched-pair analyses show that models either ignore injected candlestick semantics or move opposite to the rule-implied direction conditional on responding. The benchmark isolates a failure mode that standard observational chart benchmarks cannot detect and gives a reusable audit template for time-series imagery with controllable label mechanisms.
\end{abstract}

\section{Introduction}

Candlestick charts encode prices, volatility, and volume in a compact visual format. Modern VLMs can describe such charts fluently, but fluency does not establish financial visual understanding. A model may map an upward-sloping chart to ``bullish continuation,'' lower its probability after a risk-oriented prompt, or cite a volume pattern that did not causally affect its prediction. In financial decision support, such confident but unsupported explanations can be mistaken for signal.

The natural evaluation---scoring predictions against realized market futures---does not answer this question. Real-market prediction entangles visual understanding with noise, selection, and nonstationarity. More fundamentally, the visual features that a grounded analyst would use are often coupled to the past trend that a shortcut model extrapolates: breakouts occur at the ends of trends. We show that this is not a finite-sample inconvenience. Under strong trend--evidence coupling, the joint law of charts, labels, and responses is nearly identical for a grounded responder and a trend-shortcut responder; hence no observational evaluation at any feasible sample size separates them (\cref{thm:impossibility}). Matched evidence interventions break the tie by changing candidate candlestick evidence at fixed trend, and trend--label swaps independently test whether a score is carried by momentum. This gap between observational ambiguity and interventional identifiability is the benchmark's design principle. The same logic applies beyond finance, including monitoring waveforms, weather panels, and sensor dashboards whenever diagnostic features co-occur with global shape.

Martingale Doppelg\"anger-Eval instantiates this principle for candlestick charts using public OHLCV data and frozen models only. The benchmark has four splits (\cref{fig:schematic}): \Mzero{}, a martingale-null market that scores the same rendered chart against paired up/down labels or zero-drift stochastic-volatility futures; \Mone{}, injected local candlestick signals with matched counterfactual variants, realizing $\dop(Z_E)$; \Mtwo{}, swapped trend--label relations that decouple momentum from labels; and \Mthree{}, regime shifts that test confidence adaptation. All charts hide tickers, dates, and absolute prices; prices are normalized and rendered in multiple styles.

A common weakness of diagnostic benchmarks is that their metrics appear ad hoc. We address this by placing the audit in a structural behavioral model (\cref{sec:framework}). The model decomposes a reported probability into a null-market intercept, a trend coefficient $\beta_S$, an evidence coefficient $\beta_E$, and prompt/renderer variance components. \Mzero{} and \Mone{} identify the core structural coordinates, while \Mtwo{} supplies an over-identifying shortcut stress test under a controlled trend--label mechanism (\cref{prop:ident,prop:completeness}). We also formulate evidence faithfulness as a causal estimand with an explicit edit-geometry bias condition (\cref{prop:ces}). Because API evaluation is sequential and metered, we provide minimum detectable effects at the released budget and a block-aware e-process protocol for episode-level early stopping (\cref{sec:stats}).

Our contributions are:
\begin{enumerate}[leftmargin=1.4em,itemsep=1pt,topsep=2pt]
  \item \textbf{An identification theorem.} Observational chart evaluation cannot distinguish grounded from trend-shortcut responders under strong trend--evidence coupling, while matched evidence interventions separate them at an exponential rate; martingale-null optimality becomes a corollary (\cref{thm:impossibility,cor:m0-null}).
  \item \textbf{A structural identification map.} A behavioral response model in which \Mzero{} identifies null and trend coordinates, \Mone{} identifies evidence sensitivity, and \Mtwo{} over-identifies momentum shortcuts; the same framework gives a causal-estimand formulation of evidence faithfulness and a quantitative link to proper-scoring-rule excess risk (\cref{prop:ident,prop:completeness,prop:ces,thm:brier}).
  \item \textbf{An evaluation statistics toolkit.} Minimum detectable effects and sample-complexity bounds at the released budget, tie-aware paired testing, block-aware sequential testing for metered API evaluation, and an overlap-weighted artifact robustness check computed from existing logs (\cref{sec:stats,prop:deconfound}).
  \item \textbf{A benchmark and audited baselines.} A reproducible public shadow-market benchmark over frozen commercial and open VLMs. A single identified regression shows $|\beta_S|\gg|\beta_E|$ for every model; open models mostly ignore injected evidence, while commercial models often move in the direction opposite to the injected rule conditional on responding; the artifact check shows that these failures persist on organic-looking windows.
\end{enumerate}

\section{Related Work}

\textbf{Financial prediction benchmarks.} StockNet \citep{xu2018stocknet}, Adv-ALSTM \citep{feng2019advalstm}, and FI-2010 \citep{ntakaris2018fi2010} evaluate movement prediction against realized futures. Classical data-snooping critiques warn that apparent predictability can arise from repeated testing \citep{brock1992technical,white2000reality}. We use public OHLCV sources, but the benchmark does not target tradability: its shadow futures have labels controlled by the generator.

\textbf{Chart and multimodal understanding.} ChartBench \citep{xu2023chartbench} and CharXiv \citep{wang2024charxiv} show that chart understanding requires visual, numerical, and logical reasoning; chart-specialized pipelines include DePlot, Pix2Struct, and MatCha \citep{liu2023deplot,lee2023pix2struct,liu2023matcha}. Martingale Doppelg\"anger-Eval is complementary: it asks whether candlestick interpretation survives controlled financial counterfactuals.

\textbf{Causal identification and invariance.} Our framework adapts standard tools---structural causal models and interventions \citep{pearl2009causality}, propensity scores and overlap weighting \citep{rosenbaum1983propensity,li2018balancing,crump2009dealing}, and invariance-based identification \citep{peters2016causal,arjovsky2019invariant}---to the evaluation of a black-box responder rather than the estimation of a treatment effect. The non-identifiability argument in \cref{thm:impossibility} is a two-point Le Cam construction \citep{tsybakov2009introduction} applied to evaluation functionals.

\textbf{VLMs, shortcuts, and uncertainty.} We evaluate frozen commercial and open VLMs \citep{openai2023gpt4,openai2024gpt4o,li2024llavaonevision,bai2025qwen25vl,zhu2025internvl3}. Our metrics connect to forecast verification and calibration decompositions \citep{brier1950,murphy1973brier,degroot1983comparison,gneiting2007probabilistic}, selective classification \citep{geifman2017selective}, shortcut learning \citep{geirhos2020shortcut}, explanation auditing \citep{ribeiro2016lime}, martingale intuition \citep{fama1970efficient}, and stochastic volatility \citep{bollerslev1986garch}; the sequential protocol builds on anytime-valid, game-theoretic statistics \citep{ville1939etude,ramdas2023game,grunwald2024safe,waudbysmith2024betting}. We follow dataset documentation principles \citep{mitchell2019modelcards,gebru2021datasheets} while emphasizing controlled diagnostics over a single leaderboard score.

\section{Problem Formulation}
\label{sec:formulation}

Let $W_i=\{O_\tau,H_\tau,L_\tau,C_\tau,V_\tau\}_{\tau=t-L+1}^{t}$ be a past OHLCV window with length $L$. We normalize prices by the first close, $\widetilde P_\tau=100P_\tau/C_{t-L+1}$, normalize volume within the window, and render an image $I_i=r(W_i,s)$ under style $s$. A frozen VLM $f$ receives $I_i$ and a prompt $q$ and returns a structured prediction: direction in \{bullish, bearish, uncertain\}, probability $\probup\in[0,1]$, abstention, semantic tags, evidence regions, and a brief explanation.

For the identification analysis we coarsen the window into a scalar trend statistic $S_i$ (past momentum), rule-relevant candlestick semantic variables $Z_{E,i}$ (breakout and volume confirmation, terminal shadow geometry), and nuisance structure $N_i$, so that $I_i=r(S_i,Z_{E,i},N_i,s)$. The pair $(S,Z_E)$ is the central axis of the audit: a trend-shortcut responder is a function of $S$ alone, whereas a grounded responder remains sensitive to $Z_E$ at fixed $S$.

The benchmark distribution is $\D=\D_{\Mzero}\cup\D_{\Mone}\cup\D_{\Mtwo}\cup\D_{\Mthree}$: the image distribution is induced by public OHLCV windows and renderer randomization, while the future label mechanism is controlled by the benchmark rather than by the realized historical future. \Mone{} realizes the intervention $\dop(Z_E\!=\!z)$ at fixed $(S,N)$ via OHLCV-level edits; \Mtwo{} realizes $\dop(S\perp Y)$ and $\dop(S\,\text{anti-aligned with}\,Y)$ via label-mechanism control.

The target is not a trading strategy. The object of evaluation is a conditional response function: when charts look realistic but the label mechanism is known, does the model report uncertainty, react to injected evidence, and keep its explanation aligned with causal edits? A model that predicts realized futures in one historical period may still fail this controlled visual-understanding audit, and \cref{thm:impossibility} shows that observational data alone cannot resolve the ambiguity.

We use candlestick understanding in an operational sense that does not presume technical-analysis rules are profitable. A model that claims to analyze candlestick charts should (i) not invent directional certainty when the mechanism makes the future independent of the chart, (ii) move its probability and tags in the rule-implied direction when a rule-relevant component is edited with everything else held fixed, and (iii) cite evidence regions whose edits matter more than matched non-evidence edits. These requirements are modest, falsifiable, and independent of real-market alpha.

The main metrics are: null overconfidence $\E_{\Mzero}|\probup-1/2|$; the trend-bias index $\mathrm{TBI}=|\E[\probup\mid M_i\in Q_{\mathrm{top}}]-\E[\probup\mid M_i\in Q_{\mathrm{bottom}}]|$ over extreme momentum quintiles; injected-signal AUC and pairwise signal sensitivity (PSS, with exact ties scored as $1/2$); the shortcut gap $\mathrm{AUC}(\Mtwo_{\mathrm{align}})-\mathrm{AUC}(\Mtwo_{\mathrm{reverse}})$; prompt fragility (PFI); renderer fragility (RFI); and the causal evidence score $\mathrm{CES}=|f(I)-f(I^{\mathrm{edit}(E)})|-|f(I)-f(I^{\mathrm{edit}(\bar E)})|$, positive when editing the model's stated evidence moves the score more than a matched non-evidence edit. \Cref{sec:framework} shows that these metrics target distinct structural coordinates rather than one aggregate leaderboard score. \Cref{tab:metrics,tab:splits} (appendix) summarize the metric--split--failure-mode map.

\section{Benchmark Construction}
\label{sec:construction}

\textbf{Data sources and leakage controls.} We use public OHLCV windows from stock-prediction benchmarks and public market-data sources, with recorded source manifests, symbol lists, download times, cleaning rules, seeds, and hashes. Raw tickers, names, dates, and absolute prices are removed from images; axes use relative scales. These controls limit two leakage modes: memorizing a real ticker's future and exploiting non-chart metadata.

\textbf{Rendering and object metadata.} Each OHLCV window is rendered under several styles (red-up/green-down, green-up/red-down, hollow black-white, with/without volume, moving averages, axes, and different backgrounds). The renderer stores object-level metadata---candle bodies, wicks, volume bars, moving-average locations, rule-relevant regions---used for counterfactual editing and evidence scoring but never shown to the VLM.

\textbf{\Mzero: martingale-null.} In \Mzero-A, a single chart is paired with both virtual labels $Y^+=1$ and $Y^-=0$; in \Mzero-B, future returns follow a zero-drift stochastic-volatility process that preserves volatility clustering; in \Mzero-C, momentum buckets are label-balanced. The target behavior is calibrated uncertainty, not accuracy.

\textbf{\Mone: injected alpha.} We inject localized candlestick signals (volume-confirmed breakouts, long-shadow reversals, failed breakouts, volatility contraction--expansion, conflict cases), each with a matched counterfactual pair that changes only the relevant semantic components. This pair design is stricter than accuracy: if two charts differ only in final breakout confirmation or volume evidence, a grounded model should move its probability in the rule-implied direction.

\textbf{\Mtwo{} and \Mthree.} \Mtwo{} constructs aligned, balanced, and reverse trend-label relations; a momentum-only shortcut looks strong in the aligned split and fails in the reverse split. \Mthree{} varies the regime (trend, range, high-volatility) in which the same local pattern appears; in high-volatility regimes the desired behavior is lower coverage rather than a deterministic rule.

\begin{figure}[t]
  \centering
  \captionsetup[subfigure]{font=scriptsize,labelfont=bf,justification=centering,skip=1pt}
  \begin{subfigure}[t]{0.31\linewidth}
    \centering
    \includegraphics[width=\linewidth]{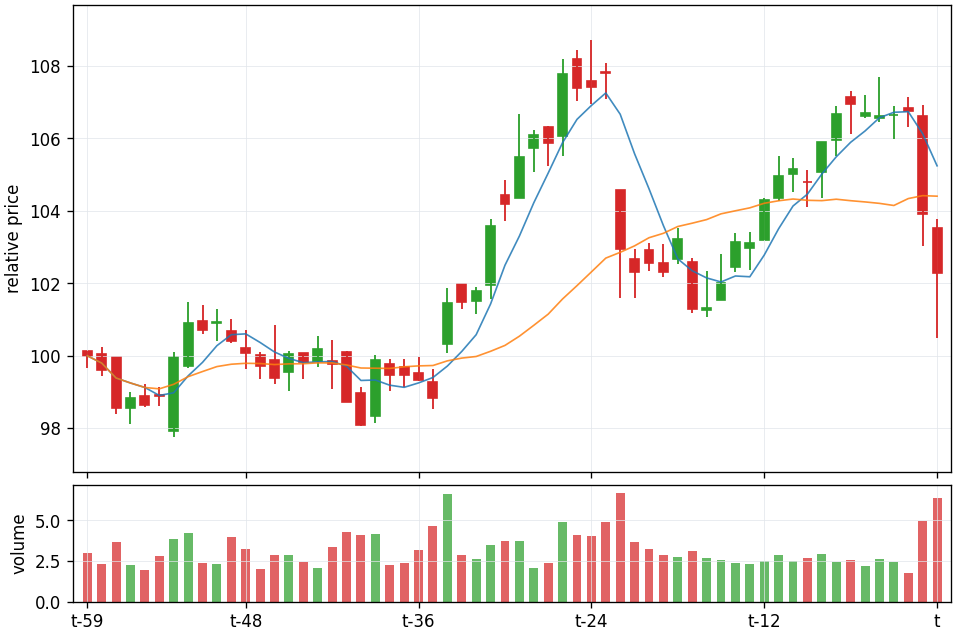}
    \caption{Observed-market future.}
  \end{subfigure}\hfill
  \begin{subfigure}[t]{0.31\linewidth}
    \centering
    \includegraphics[width=\linewidth]{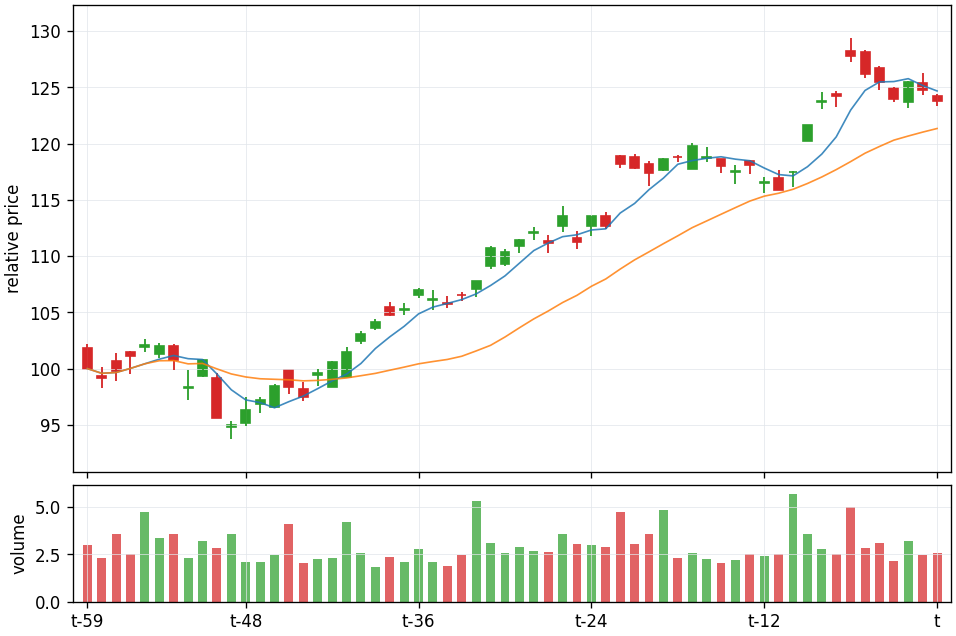}
    \caption{\Mzero{}: null labels.}
  \end{subfigure}\hfill
  \begin{subfigure}[t]{0.31\linewidth}
    \centering
    \includegraphics[width=\linewidth]{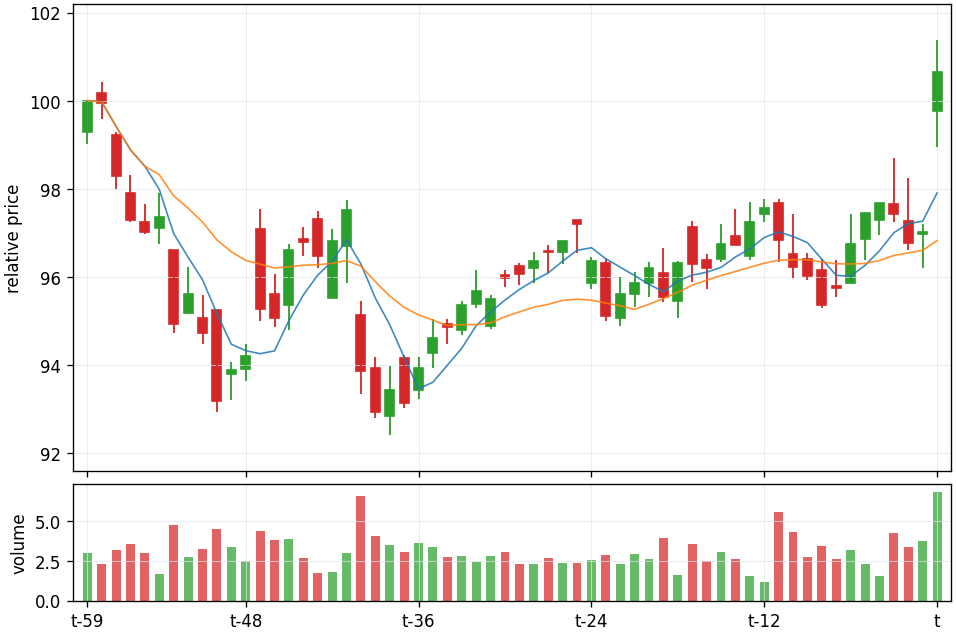}
    \caption{\Mone{}: injected evidence.}
  \end{subfigure}

  \vspace{0.25em}

  \begin{subfigure}[t]{0.39\linewidth}
    \centering
    \includegraphics[width=\linewidth]{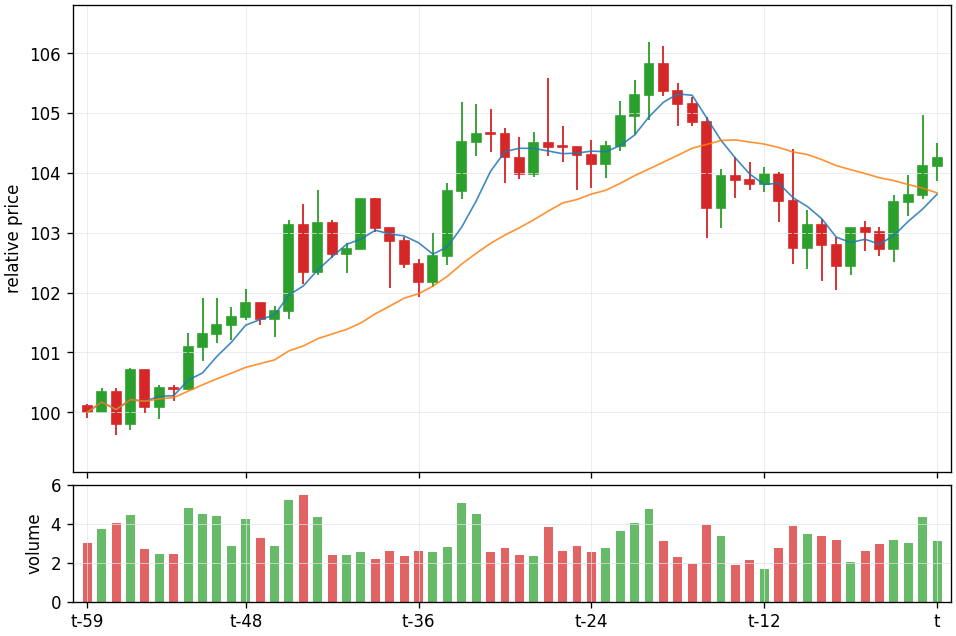}
    \caption{\Mtwo{}: trend--label swap.}
  \end{subfigure}\hspace{0.07\linewidth}
  \begin{subfigure}[t]{0.39\linewidth}
    \centering
    \includegraphics[width=\linewidth]{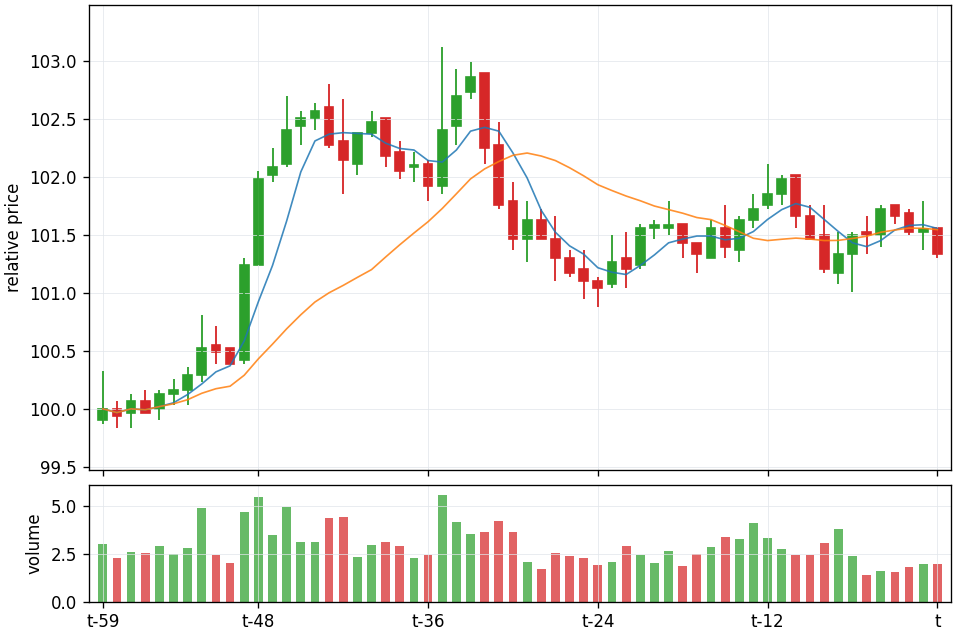}
    \caption{\Mthree{}: regime shift.}
  \end{subfigure}
  \caption{Shadow-market label mechanisms. The visible chart grammar is held fixed while the future-label mechanism changes: real-market validation uses realized futures; \Mzero{} balances or simulates null futures; \Mone{} edits localized candlestick evidence; \Mtwo{} controls the trend--label relation; and \Mthree{} changes the regime in which similar local patterns appear.}
  \label{fig:schematic}
  \vspace{0.45em}
\end{figure}

\textbf{Prompt protocol and parsing.} All models receive the same structured prompt family requesting a compact JSON object (direction, $\probup$, confidence, abstention, semantic tags, evidence regions, short explanation). We avoid chain-of-thought because the benchmark evaluates observable decisions. Variants include directional, risk-framed, technical-tagging, abstention-encouraging, rule-card, and counter-prompts. A conservative parser scores a prediction only when direction, $\probup$, and abstention are unambiguous; invalid outputs remain in raw logs but are not coerced into probabilities.

\textbf{Counterfactual edits.} The matched pairs in \Mone{} and the evidence edits for CES are performed at the OHLCV level before rendering, not by masking pixels, so the model is never shown a gray patch or blur. For instance, a volume-confirmed breakout becomes a failed breakout by lowering the final close below the range boundary and reducing the final volume bar, while older candles and style remain fixed; non-evidence controls edit a visually comparable but semantically irrelevant segment with matched magnitude (\cref{asm:edits}). CES is therefore a paired intervention on chart semantics rather than a generic saliency score.

\textbf{Artifact release.} The reproducibility package contains source manifests, generation and renderer code, prompt banks, parsers, and evaluation scripts; when a raw market source cannot be redistributed, the release provides indices, hashes, and scripts instead (\cref{app:protocol}).

\section{An Identification Framework for Chart-Understanding Audits}
\label{sec:framework}

All proofs appear in \cref{app:proofs}.

\subsection{Observational evaluation cannot identify grounding}
\label{sec:impossibility}

Call a response function $f$ mapping images to $[0,1]$ a trend-shortcut responder if $f(r(S,Z_E,N,s))=h(S)$ for some $h$, and a grounded responder if it is strictly monotone in $Z_E$ at every fixed $(S,N,s)$. An observational audit of length $m$ draws $(I_1,Y_1),\dots,(I_m,Y_m)$ from a market distribution $\Dobs$, queries the model once per image, and decides whether the responder is grounded or shortcut-driven from the transcript $\{(I_i,Y_i,f(I_i))\}_{i\le m}$.

\begin{theorem}[Observational non-identifiability; interventional identifiability]
\label{thm:impossibility}
For every coupling level $\varepsilon\in(0,1)$ there exist a market distribution $\Dobs^{\varepsilon}$, a grounded responder $f_g$, and a trend-shortcut responder $f_s$ such that:
\begin{enumerate}[leftmargin=1.6em,itemsep=0pt,topsep=2pt]
\item[(a)] \textbf{(Impossibility.)} $\tv\big(\mathcal{L}_{\Dobs^\varepsilon}^{m}(f_g),\,\mathcal{L}_{\Dobs^\varepsilon}^{m}(f_s)\big)\le \min\{1,m\varepsilon\}$, so every observational audit of length $m$ has worst-case error at least $\tfrac12(1-m\varepsilon)_+$; as the coupling becomes deterministic ($\varepsilon\to0$), no observational audit of any fixed length beats chance.
\item[(b)] \textbf{(Identification under intervention.)} Under the \Mone{} design---$n$ matched pairs realizing $\dop(Z_E\!=\!1)$ vs.\ $\dop(Z_E\!=\!0)$ at shared $(S,N,s)$---the tie-aware paired sign statistic separates $f_g$ from $f_s$ with error at most $\exp(-n\delta^2/2)$, where $\delta>0$ is the induced hit-rate margin of $f_g$. Under an \Mtwo{} design in which non-trend evidence is balanced across labels, a momentum-only shortcut is separated from an $S$-invariant responder by the sign of its aligned--reverse AUC gap at the same exponential rate.
\end{enumerate}
\end{theorem}

The construction places breakout evidence at trend ends with probability $1-\varepsilon$, matching the dependence that motivates the audit; the proof is a two-point argument \citep{tsybakov2009introduction} on the induced transcript laws. The theorem turns the benchmark design from an empirical preference into an identification requirement. It also limits what leaderboard-style chart benchmarks can claim: no score computed under a fixed observational chart--label joint can certify evidence use.

\begin{corollary}[Null optimality of paired \Mzero{} labels]
\label{cor:m0-null}
For \Mzero-A, suppose a model is evaluated by any strictly proper binary scoring rule on the paired labels $\{(I_i,1),(I_i,0)\}$ with equal weight. The unique Bayes-optimal probability for every image is $\probup=1/2$, and abstention (scored as $p=1/2$) is equally optimal. Consequently, any systematic dependence of $\probup$ on past momentum in \Mzero{} is a property of the model, not evidence of directional alpha under the benchmark mechanism.
\end{corollary}

The next result translates the two \Mzero{} diagnostics into proper-score regret, in the spirit of the Murphy decomposition \citep{murphy1973brier,degroot1983comparison}.

\begin{theorem}[Calibration--decision link]
\label{thm:brier}
On \Mzero-A with quintile buckets ($q=0.2$ mass each), the excess Brier risk of a responder $p$ over the optimal report $1/2$ satisfies
\[
  \E\,B(p,Y)-B(\tfrac12,Y)\;=\;\E\big[(p-\tfrac12)^2\big]\;\ge\;\max\Big\{\mathrm{OverConf}^2,\ \tfrac{q}{2}\,\mathrm{TBI}^2\Big\},
\]
and under log loss the excess risk is at least $2\max\{\mathrm{OverConf}^2,(q/2)\,\mathrm{TBI}^2\}$.
\end{theorem}

For example, the observed $\mathrm{TBI}=0.32$ and $\mathrm{OverConf}=0.22$ of Qwen2.5-VL-7B force excess Brier risk of at least $0.010$ and $0.050$ respectively, against a maximum possible excess of $0.25$: each diagnostic lower-bounds a regret that any downstream user of the probabilities would pay.

\subsection{A structural behavioral model and what each split identifies}
\label{sec:structural}

We use a working response model on the link scale. For image $i$ rendered in style $s$ and queried with prompt $q$,
\begin{equation}
\label{eq:structural}
  \sigma^{-1}(\probup^{(i,q,s)})\;=\;\alpha+\beta_S\,\widetilde S_i+\beta_E\,E_i+\gamma_q+\gamma_s+\varepsilon_i ,
\end{equation}
where $\widetilde S_i$ is standardized past momentum, $E_i\in\{-1,0,+1\}$ is the injected-evidence direction ($0$ when none), $\gamma_q,\gamma_s$ are prompt and style effects, and $\varepsilon_i$ is mean-zero response noise. We do not assume that \cref{eq:structural} is a generative model for any VLM. Its role is to define audit coordinates in which shortcut use, evidence sensitivity, anchoring, and fragility are distinct parameters.

\begin{proposition}[Identification map]
\label{prop:ident}
Under \cref{eq:structural} with $\E[\widetilde S]=0$ and the benchmark designs:
\begin{enumerate}[leftmargin=1.6em,itemsep=0pt,topsep=2pt]
\item[(i)] \Mzero{} identifies $(\alpha,\beta_S)$ on the link scale: $\alpha=\E[\sigma^{-1}(p)]$, $\beta_S=\operatorname{Cov}(\sigma^{-1}(p),\widetilde S)/\operatorname{Var}(\widetilde S)$. The link-scale quintile contrast is proportional to $|\beta_S|$; the probability-scale TBI used in tables is a sign-aligned diagnostic under sign-separated buckets and independent noise. OverConf satisfies $\E|p-1/2|\ge|\E p-1/2|$.
\item[(ii)] The \Mone{} paired contrast identifies $\beta_E$ free of $(\alpha,\beta_S,\gamma)$ because pair members share $(S,N,q,s)$:
\[
\sigma^{-1}(p^{\mathrm{sig}})-\sigma^{-1}(p^{\mathrm{ctr}})
=\beta_E\,\Delta E+\Delta\varepsilon .
\]
the tie-aware PSS is the sign test of $\beta_E$, invariant to any strictly monotone recalibration of $p$.
\item[(iii)] For a momentum-only score with independent noise and balanced non-trend evidence, \Mtwo{} over-identifies trend sensitivity: the aligned--reverse AUC gap has the sign of the momentum effect under non-degenerate trend--label coupling.
\item[(iv)] PFI and RFI identify $\operatorname{Var}(\gamma_q)$ and $\operatorname{Var}(\gamma_s)$, up to query noise, from paired prompt and renderer variation.
\end{enumerate}
\end{proposition}

\begin{proposition}[No single split suffices; \Mtwo{} is an over-identifying stress test]
\label{prop:completeness}
For each single split there exist two responders of the form \cref{eq:structural} with different $(\beta_S,\beta_E)$ that induce identical metric values on that split: \Mzero{} cannot separate $\beta_E>0$ from $\beta_E=0$; the \Mone{} paired contrast cannot separate $(\alpha,\beta_S)$ values; \Mtwo{} cannot separate evidence use from trend use when incidental evidence covaries with trend. Under the matched-pair assumptions, \Mzero{}+\Mone{} identify $(\alpha,\beta_S,\beta_E)$; \Mtwo{} is not needed for point identification but provides an over-identifying diagnostic for momentum shortcuts under a controlled trend--label mechanism.
\end{proposition}

\Cref{prop:completeness} answers why the benchmark needs several splits rather than one score, without claiming that every split is necessary for the same parameter vector. A single aggregate necessarily discards identified structure; \Mtwo{} instead tests whether the trend coordinate behaves as a shortcut under a flipped mechanism. The classical trend-confounder statement---a score that is a monotone function of past momentum alone has aligned AUC above reverse AUC---appears within this map as \cref{prop:swap} (appendix). Estimating \cref{eq:structural} directly is also informative: a regression of parsed probabilities on momentum, injected-signal direction, and prompt/style dummies, computable from existing logs, summarizes the benchmark in one identified table (\cref{tab:joint-regression}).

\subsection{Evidence faithfulness as a causal estimand}
\label{sec:ces}

Write the model score as $f(I)=g(Z_E,Z_C,Z_R)$, where $Z_E$ are the semantic variables in the evidence region the model itself reports, $Z_C$ a matched control region, and $Z_R$ the remainder. The CES edits realize interventions $I^{E}=r(\dop(Z_E\!\to\!T_E Z_E))$ and $I^{C}=r(\dop(Z_C\!\to\!T_C Z_C))$. The estimand rests on two verifiable conditions (\cref{asm:edits}, appendix): matched magnitude---the rendered displacements $d(I,I^E)$ and $d(I,I^C)$ agree in a fixed visual metric $d$ along the same edit path (the released generator matches glyph-level $\ell_2$ displacement)---and locality---each edit alters only the glyphs rendering its target variables.

\begin{proposition}[CES as a causal contrast]
\label{prop:ces}
Define the evidence-sensitivity gap
\[
\kappa=\E\,\big|g(T_EZ_E,Z_C,Z_R)-g(Z_E,Z_C,Z_R)\big|-\E\,\big|g(Z_E,T_CZ_C,Z_R)-g(Z_E,Z_C,Z_R)\big| .
\]
Under \cref{asm:edits}, $\E[\mathrm{CES}]=\kappa$; hence $\mathrm{CES}>0$ in expectation iff the model's stated evidence is operationally faithful ($\kappa>0$). If exact magnitude matching fails, let $\widetilde I_i^{E}$ denote the evidence edit rescaled along the same edit path to match $d(I_i,I_i^{C})$, and define $\Delta_i^{\mathrm{path}}=d(I_i^{E},\widetilde I_i^{E})$. If $f$ is $L$-Lipschitz with respect to $d$, then $\big|\E[\mathrm{CES}]-\kappa\big|\le L\,\E\Delta^{\mathrm{path}}$.
\end{proposition}

\Cref{prop:ces} separates two interpretations of a near-zero CES: an unfaithful explanation ($\kappa\approx0$) and an underpowered test ($|\kappa|$ below the minimum detectable effect). \Cref{prop:mde} resolves the distinction at the released budget. For open models evaluated on $n=2000$ edit pairs, the MDE is ${\sim}0.004$, so their CES values are measured zeros; for API models evaluated on $n=300$ pairs, the MDE is ${\sim}0.011$, comparable to the observed values.

\subsection{Robustness as group invariance}
\label{sec:invariance}

Let $G=G_r\times G_q$ be the semantic-preserving transformation group: $G_r$ acts on render style (color convention, axes, grid, background, volume panel) and $G_q$ acts on prompt phrasing, with neither changing $(W,Z_E)$. Any responder that is measurable in $(S,Z_E,N)$ alone---including any grounded responder---is constant on $G$-orbits. RFI and PFI are the orbit-variance functionals of $G_r$ and $G_q$ and vanish exactly on $G$-invariant responders; the order-two color subgroup (red-up $\leftrightarrow$ green-up) yields the direct paired test $\E[p\circ c-p]$ for color-heuristic dependence (\cref{prop:invariance}, appendix). This connects the fragility metrics to invariance-based identification \citep{peters2016causal,arjovsky2019invariant}: failures of $G$-invariance certify dependence on presentation rather than candle geometry. The color test is a zero-cost audit on existing logs; all evaluated models pass it (\cref{sec:results}), localizing the observed failure to slope geometry rather than color.

\subsection{Artifact robustness via the generator classifier}
\label{sec:deconfound}

Synthetic benchmarks face a standard objection: a generator classifier distinguishes \Mone{} from \Mzero{} windows with AUC $0.93$ (\cref{tab:audit}), so models might respond to synthetic fingerprints rather than financial semantics. We convert this disclosed risk into a sensitivity estimator. Let $X$ be the background features that make the generator classifiable and $e(X)=\Prb(\Mone\mid X)$ the fingerprint propensity.

\begin{proposition}[Overlap-weighted artifact robustness]
\label{prop:deconfound}
Assume overlap: $e(X)<1$ on a set of positive \Mone{} mass. Weighting \Mone{} pairs by $w(X)=1-e(X)$ identifies the total pair contrast in the overlap population $P_{\mathrm{ow}}(dx)\propto e(x)(1-e(x))P(dx)$, the subpopulation of injected windows whose recorded backgrounds are organic-looking \citep{rosenbaum1983propensity,li2018balancing}. This is not, by itself, a semantic-contrast identification result: if the pair contrast is $D=\zeta(X)+a(X)+\Delta\varepsilon$, the estimand is $\E_{P_{\mathrm{ow}}}[\zeta(X)+a(X)]$. It equals the semantic contrast only under an additional restriction such as pair-invariant artifacts, $\E_{P_{\mathrm{ow}}}a(X)=0$, or an organic-control contrast that subtracts $a(X)$.
\end{proposition}

The adjusted estimates are recomputable from existing logs at zero API cost (\cref{tab:artifact-adjusted}). Their role is sensitivity analysis: if near-random \Mone{} recovery were driven only by recorded background imbalance, moving to the organic-looking overlap population should materially change the estimates, which \cref{sec:results} tests directly.

\subsection{Statistical toolkit: power, sequential validity, and dependence}
\label{sec:stats}

\textbf{Minimum detectable effects.}
\begin{proposition}[MDEs and sample complexity at the released budget]
\label{prop:mde}
At level $\alpha=0.05$ (two-sided) and power $0.8$: the tie-aware paired sign test on $n$ \Mone{} pairs detects $|\mathrm{PSS}-\tfrac12|\ge 1.40/\sqrt{n}$ ($0.075$ at $n=351$); the quintile contrast detects $\mathrm{TBI}\ge 2.80\,\sigma_p\sqrt{2/(qn)}$ ($0.044$ at $n=800$, $\sigma_p=0.14$); the paired $t$-test detects $|\mathrm{CES}|\ge 2.80\,\sigma_{\mathrm{CES}}/\sqrt{n}$ ($0.011$ at $n=300$, $0.004$ at $n=2000$, $\sigma_{\mathrm{CES}}\approx0.07$). Conversely, any test distinguishing $\beta_E=0$ from $|\beta_E|=\delta$ with error $\le\alpha$ requires $n\ge c\log(1/\alpha)/\delta^{2}$ matched pairs.
\end{proposition}

All headline TBI effects exceed their MDEs. The tie-aware PSS estimates also make the M1 failure interpretable as either non-response (high tie rate) or opposite-direction response (low conditional sign accuracy). The only estimates near their MDE are the API-model CES values, which \cref{prop:ces} already treats as power-limited.

\textbf{Anytime-valid sequential evaluation.} API evaluation is metered, and stopping early after an interim look invalidates fixed-$n$ inference. For tie-aware pair scores $H_k\in\{0,\tfrac12,1\}$ with null $\E[H_k\mid\mathcal F_{k-1}]=\tfrac12$, the betting process $E_n=\prod_{k\le n}(1+\lambda_k(H_k-\tfrac12))$ with predictable $\lambda_k\in(-2,2)$ is a nonnegative supermartingale, so $\Prb(\exists n: E_n\ge1/\alpha)\le\alpha$ by Ville's inequality \citep{ville1939etude,ramdas2023game,grunwald2024safe,waudbysmith2024betting}. Under within-episode dependence, validity comes from betting on episode-level scores or block-level e-values, not from randomizing pair order; randomization only avoids replay-order artifacts (\cref{app:evalues}). Each metric is also a plug-in functional of paired contrasts whose influence functions are asymptotically normal under episode-block dependence (\cref{app:asymptotics}), giving the released block-bootstrap and FDR protocol \citep{kunsch1989block,benjamini1995fdr} an estimation-theoretic basis.

\section{Experiments}

\textbf{Data and rendering.} The benchmark uses public OHLCV sources and benchmark seeds: StockNet/ACL18-style equities \citep{xu2018stocknet}, ACL18/KDD17 splits \citep{feng2019advalstm}, public daily equity/ETF data, crypto klines, and FI-2010 for optional high-frequency aggregation. The main setting uses $L=60$, $H=5$ (robustness: $L\in\{30,60,120\}$, $H\in\{1,5,10,20\}$), with render variants for color convention, axes, volume, moving averages, background, and aspect ratio.

\textbf{Models.} We report completed frozen VLM evaluations for commercial API models and open-weight VLMs. Commercial endpoints are treated as black boxes; released manifests record the configured API model identifier and evaluation window for each run. Runs still in progress at draft time are excluded from claims.

\textbf{Evaluation.} The main API-budget split evaluates 800 \Mzero, 800 \Mone, 400 \Mtwo, and 400 \Mthree{} samples per completed model; prompt and renderer subsets use paired designs. Confidence intervals use a block bootstrap over asset/date episodes with FDR control across comparisons \citep{kunsch1989block,benjamini1995fdr}, justified asymptotically in \cref{app:asymptotics}. Traditional baselines (\cref{tab:rules}) include random prediction, abstention, momentum continuation, mean reversion, moving-average crossover, RSI-style reversal, breakout rules, and a LightGBM technical-indicator diagnostic. The last is not a proposed model; it checks that the benchmark is solvable from explicit generator features.

\textbf{Interpretation protocol.} Two conservative rules govern conclusions. First, we make claims only for completed model--split pairs. Second, a model has not recovered a candlestick signal unless it beats both chance and the corresponding trend baseline, because many injected patterns sit at trend ends.

\begin{table}[t]
\centering
\small
\setlength{\tabcolsep}{4.2pt}
\caption{Main frozen-model results on Martingale Doppelg\"anger-Eval. M0 is null-market calibration; M1 tests injected-signal recovery; M2 tests trend-confounder shortcuts. PSS is tie-aware, scoring exact ties as $1/2$; lower is better for OverConf, TBI, and Gap; higher is better for M1 AUC/PSS and CES.}
\label{tab:main-results}
\begin{tabular*}{\linewidth}{@{\extracolsep{\fill}}lrrrrrrr@{}}
\toprule
Model & OverConf$\downarrow$ & TBI$\downarrow$ & $\rho(p,M)$ & M1 AUC$\uparrow$ & PSS$\uparrow$ & M2 Gap$\downarrow$ & CES$\uparrow$ \\
\midrule
GPT-5.3 & 0.101 & 0.133 & 0.426 & 0.512 & 0.430 & 0.503 & -0.003 \\
GPT-4.1 & 0.103 & 0.174 & 0.467 & 0.517 & 0.466 & 0.455 & 0.007 \\
Claude Sonnet API & 0.161 & 0.166 & 0.399 & 0.488 & 0.444 & 0.448 & 0.014 \\
GPT-4o & 0.088 & 0.145 & 0.395 & 0.486 & 0.453 & 0.410 & 0.013 \\
Qwen2.5-VL-7B & 0.224 & 0.322 & 0.549 & 0.486 & 0.504 & 0.520 & -0.018 \\
LLaVA-OV-7B & 0.036 & 0.113 & 0.619 & 0.469 & 0.507 & 0.399 & -0.004 \\
InternVL3-8B & 0.192 & 0.300 & 0.586 & 0.478 & 0.499 & 0.591 & -0.003 \\
\bottomrule
\end{tabular*}
\end{table}

\begin{figure}[H]
  \centering
  \includegraphics[width=0.86\linewidth]{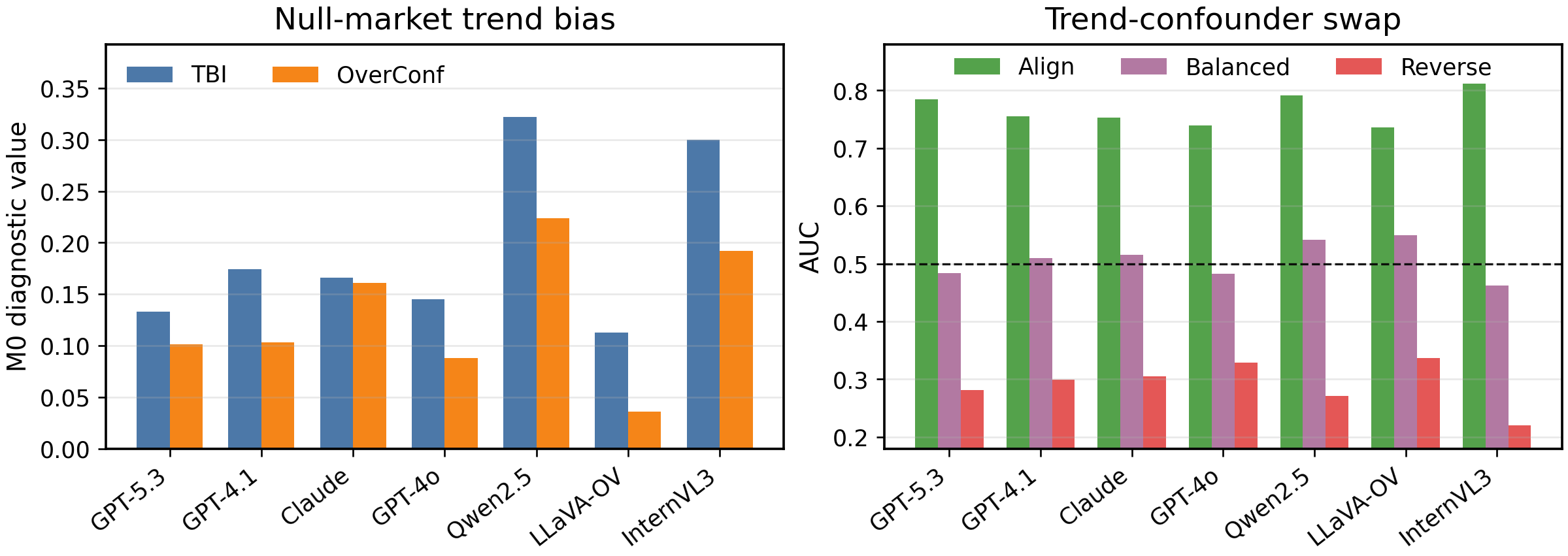}
  \caption{Core diagnostics. Left: M0 trend-bias index and null overconfidence. Right: M2 AUC under aligned, balanced, and reverse trend-label relations (dashed line: chance). Large aligned--reverse separation indicates trend-following shortcuts.}
  \label{fig:core}
  \vspace{0.5em}
\end{figure}

\section{Results}
\label{sec:results}

\textbf{One identified table: trend dominates evidence for every model.} \Cref{tab:joint-regression} estimates \cref{eq:structural} directly on all parsed \Mzero/\Mone{} responses from existing logs. Every frozen model has a strongly positive trend coefficient ($\beta_S$ from $0.20$ to $0.41$ on the log-odds scale, all $|t|>8$). By contrast, $\beta_E$ is either indistinguishable from zero (open models) or significantly negative (all four commercial APIs; e.g.\ Claude Sonnet API $-0.17$, SE $0.05$). The negative sign is diagnostic rather than incidental: the dominant injected rule family is the low-volume failed breakout, where surface motion is upward but rule semantics are bearish, and pooled PSS on that family is $0.12$. The models track the surface direction of the final move rather than the rule semantics; the shortcut operates even at the pattern level.

\begin{table}[t]
\centering
\small
\setlength{\tabcolsep}{4.5pt}
\caption{Structural behavioral regression estimated from existing logs: $\mathrm{logit}(\probup)=\alpha+\beta_S S+\beta_E E+\gamma^\top(\text{style},\text{prompt})$ on all parsed \Mzero{}/\Mone{} responses, with $S$ standardized past momentum and $E\in\{-1,0,+1\}$ the injected-signal direction. Standard errors are clustered by base window. Every frozen model loads heavily on the trend shortcut ($\beta_S$) and only weakly, if at all, on injected candlestick evidence ($\beta_E$); the trained positive control reverses this ordering.}
\label{tab:joint-regression}
\begin{tabular*}{\linewidth}{@{\extracolsep{\fill}}lrrrrr@{}}
\toprule
Model & $n$ & $\beta_S$ (SE) & $\beta_E$ (SE) & $|\beta_S/\beta_E|$ & $\max|\gamma_{\text{prompt}}|$ \\
\midrule
GPT-5.3 & 6237 & 0.200 (0.021) & -0.077 (0.028) & 2.6 & 0.097 \\
GPT-4.1 & 6700 & 0.268 (0.027) & -0.077 (0.030) & 3.5 & 0.121 \\
Claude Sonnet API & 4550 & 0.372 (0.039) & -0.171 (0.046) & 2.2 & 0.187 \\
GPT-4o & 6700 & 0.216 (0.020) & -0.064 (0.028) & 3.4 & 0.159 \\
Qwen2.5-VL-7B & 6693 & 0.409 (0.037) & 0.008 (0.034) & 49.9 & 0.916 \\
LLaVA-OV-7B & 6691 & 0.208 (0.025) & -0.020 (0.015) & 10.5 & 0.186 \\
InternVL3-8B & 6694 & 0.365 (0.035) & 0.042 (0.027) & 8.8 & 1.065 \\
\bottomrule
\end{tabular*}
\end{table}

\textbf{Frozen VLMs are trend-sensitive in the null market.} \Cref{tab:main-results,fig:core} show that completed frozen VLMs are not calibrated to the martingale-null mechanism: M0 TBI ranges from $0.113$ to $0.322$, and the Spearman correlation between $\probup$ and past momentum is positive for every model. By \cref{cor:m0-null}, this dependence is a model property rather than alpha; by \cref{thm:brier}, it is costly, forcing excess Brier risk up to $0.050$; and by \cref{prop:mde}, all TBI estimates exceed the $0.044$ MDE.

\textbf{Injected signals are not recovered---and the failure has two distinct forms.} M1 AUC is near random for all frozen VLMs ($0.469$--$0.517$), and tie-aware PSS never materially exceeds chance. The matched-pair decomposition separates two mechanisms: open models mostly do not respond (86--96\% of pairs yield numerically identical probabilities), while commercial APIs respond more often (tie rates 15--67\%) but move opposite to the rule-implied direction conditional on responding (sign accuracy $0.36$--$0.46$), consistent with the negative $\beta_E$. Both modes are far from the trained positive controls, which reach PSS $0.91$ on the same pairs (\cref{tab:dcc-controls}). The benchmark therefore rewards grounding when it exists rather than merely penalizing all models.

\textbf{The recorded-artifact objection weakens under overlap adjustment.} Generator distinguishability (AUC $0.93$, \cref{tab:audit}) raises the possibility that models respond to synthetic fingerprints. Applying \cref{prop:deconfound}, \cref{tab:artifact-adjusted} reweights pairs toward organic-looking recorded backgrounds. Recovery remains weak, so the finding is stable in the overlap population, although this analysis does not remove unrecorded artifacts.

\begin{table}[t]
\centering
\small
\setlength{\tabcolsep}{4pt}
\caption{Overlap-weighted \Mone{} robustness. Weights $w=1-\hat e$ target injected pairs with organic-looking recorded backgrounds (propensity AUC 0.892). PSS is tie-aware; ESS is Kish effective sample size; brackets are pair-bootstrap 95\% intervals.}
\label{tab:artifact-adjusted}
\begin{tabular*}{\linewidth}{@{\extracolsep{\fill}}lrrrrrr@{}}
\toprule
Model & Pairs & ESS & PSS & PSS$^{\mathrm{adj}}$ & AUC & AUC$^{\mathrm{adj}}$ \\
\midrule
GPT-5.3 & 351 & 208 & 0.430 & 0.405 [0.36, 0.46] & 0.512 & 0.316 [0.22, 0.43] \\
GPT-4.1 & 351 & 208 & 0.466 & 0.454 [0.39, 0.51] & 0.517 & 0.323 [0.23, 0.44] \\
Claude Sonnet API & 351 & 208 & 0.444 & 0.427 [0.39, 0.47] & 0.488 & 0.308 [0.23, 0.42] \\
GPT-4o & 351 & 208 & 0.453 & 0.431 [0.40, 0.46] & 0.486 & 0.324 [0.24, 0.43] \\
Qwen2.5-VL-7B & 347 & 206 & 0.504 & 0.508 [0.49, 0.53] & 0.482 & 0.313 [0.24, 0.41] \\
LLaVA-OV-7B & 351 & 208 & 0.507 & 0.509 [0.49, 0.53] & 0.469 & 0.343 [0.26, 0.44] \\
InternVL3-8B & 351 & 208 & 0.499 & 0.481 [0.45, 0.51] & 0.478 & 0.330 [0.25, 0.43] \\
\bottomrule
\end{tabular*}
\end{table}

\textbf{Trend-confounder swaps expose shortcut use---in both directions.} M2 shortcut gaps range from $0.400$ to $0.591$ (open VLMs) and $0.410$ to $0.503$ (commercial APIs); the right panel of \cref{fig:core} shows the pattern predicted by \cref{prop:swap}. The fully fine-tuned positive control DCC-VLM-Full has an M2 gap of $-0.465$ (\cref{tab:dcc-controls}): aggressive optimization for injected evidence produced a reverse shortcut. The calibrated control DCC-VLM-Full-Cal approaches the desirable corner (gap $-0.019$, PSS $0.83$).

\begin{figure}[H]
  \centering
  \includegraphics[width=0.72\linewidth]{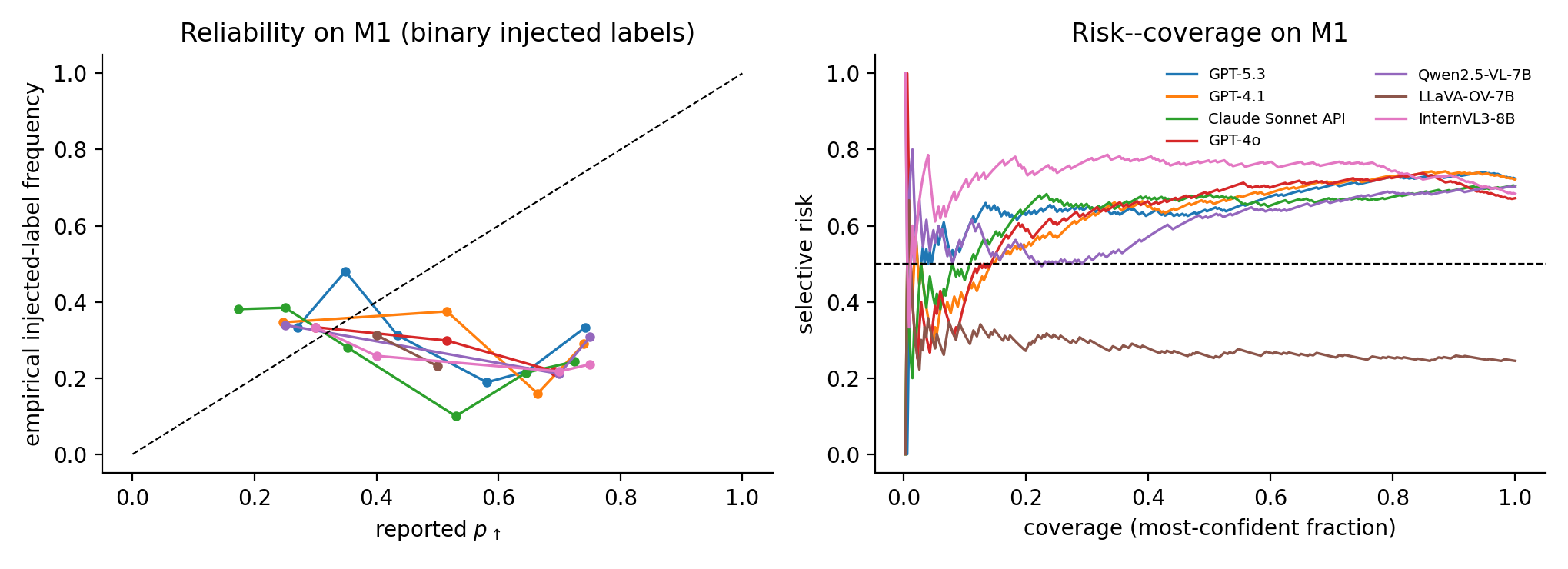}
  \caption{Reliability and selective risk on \Mone{} pairs; confidence is often anti-informative.}
  \label{fig:reliability}
  \vspace{-0.5em}
\end{figure}

\textbf{Evidence faithfulness is weak.} \Cref{fig:reliability} shows flat-to-inverted reliability curves, and CES values are close to zero, with several open models negative. By \cref{prop:ces,prop:mde}, open-model estimates are measured zeros ($n=2000$, MDE $0.004$), while API estimates cannot rule out small positive faithfulness ($n=300$, MDE $0.011$).

\textbf{Models pass the color-invariance audit.} The color-subgroup test of \cref{prop:invariance} gives mean paired shifts $|\E[p\circ c-p]|\le0.005$ (all Wilcoxon $p\ge0.12$), ruling out a red/green heuristic and localizing the failure to slope geometry.

\section{Discussion and Limitations}

\textbf{Scope of the audit template.} The benchmark instantiates a general recipe---martingale-null labels, injected counterfactual evidence, and confounder swaps---for time-series imagery. The identification analysis uses no financial structure beyond a trend statistic coupled to local evidence, so the same logic applies to waveforms, weather charts, and sensor dashboards.

\textbf{Limitations.} The benchmark does not claim that candlestick patterns are profitable or tradable. \Cref{prop:deconfound} is an overlap-population robustness check for recorded artifacts, not a removal of unrecorded or pair-differential fingerprints. The structural model in \cref{eq:structural} omits interactions, and commercial API results identify black-box systems only by manifest strings and evaluation windows. The benchmark audits single-chart understanding rather than portfolios, execution, news, or cross-asset context.

\section{Conclusion}

We presented Martingale Doppelg\"anger-Eval, a shadow-market benchmark showing that observational chart evaluation cannot certify evidence use, while matched interventions and trend--label swaps can audit evidence sensitivity and momentum shortcuts. Frozen VLMs are trend-biased in null markets, weak on injected candlestick evidence, vulnerable to trend-confounder swaps, and only weakly faithful to stated evidence. Each conclusion is tied to an identified coordinate, a released power calculation, and explicit artifact-robustness assumptions.

\label{page:lastmain}
\clearpage
\bibliographystyle{plainnat}
\bibliography{references}

\clearpage
\appendix
\section{Proofs and Extended Theoretical Analysis}
\label{app:proofs}

\subsection{Notation and Scoring Conventions}

For a base window $W_i$, let $I_i=r(W_i,s_i)$ denote the rendered image under
style $s_i$, and let $p_i=f(I_i,q_i)$ be the model's reported probability of an
up move under prompt $q_i$. The benchmark evaluates a forecast using a loss
$\ell(p,y)$ for $p\in[0,1]$ and $y\in\{0,1\}$. A scoring rule is strictly proper
if, for any event probability $\pi\in[0,1]$, the expected score
\[
  R_\pi(p)=\pi \ell(p,1)+(1-\pi)\ell(p,0)
\]
is uniquely minimized at $p=\pi$. Log loss and the Brier score are the two
canonical examples used in our analysis and implementation.

The abstention option is treated as a separate action. If abstention receives a
fixed neutral score equal to the loss of reporting $p=1/2$, both abstention and
$p=1/2$ are Bayes optimal in \Mzero-A. If abstention is penalized more heavily,
the unique probabilistic optimum remains $p=1/2$. Thus the null-market results
should be read as statements about the probability forecast: directional
confidence is not justified by the benchmark mechanism.

\subsection{Proof of \texorpdfstring{\Cref{thm:impossibility}}{Theorem 1} (Observational Non-Identifiability)}
\label{app:impossibility-proof}

\begin{proof}
\textbf{Construction.} Fix $\varepsilon\in(0,1)$, a threshold $\tau$, and two
probability values $0<p_0<p_1<1$ with margin $\delta_0=(p_1-p_0)/2$. Define the
observational market $\Dobs^{\varepsilon}$ as follows. The trend statistic $S$
has a continuous distribution; the binary evidence variable is
\[
  Z_E \;=\; \ind\{S>\tau\}\oplus B,\qquad B\sim\mathrm{Bernoulli}(\varepsilon)
  \ \text{independent of } (S,N,Y),
\]
where $\oplus$ is exclusive-or; nuisance $N$ and the label $Y$ have arbitrary
joint distribution with $(S,N)$. Thus with probability $1-\varepsilon$ the
evidence indicator agrees exactly with the trend-end indicator---the
``breakouts happen at trend ends'' regime that real markets approximate.
The renderer $r$ is injective in $(S,Z_E,N)$ on its support, so both quantities
below are well-defined functions of the image. Define
\[
  f_s(I)=p_0+(p_1-p_0)\,\ind\{S(I)>\tau\},\qquad
  f_g(I)=p_0+(p_1-p_0)\,Z_E(I).
\]
$f_s$ is a trend-shortcut responder and $f_g$ is a grounded responder with
response margin $p_1-p_0=2\delta_0$.

\textbf{(a) Impossibility.} Under $\Dobs^{\varepsilon}$,
$\Prb(f_g(I)\neq f_s(I))=\Prb(B=1)=\varepsilon$. An observational transcript of
length $m$ is $T_m(f)=\{(I_i,Y_i,f(I_i))\}_{i\le m}$ with the $(I_i,Y_i)$ drawn
i.i.d.\ from $\Dobs^{\varepsilon}$. Couple the two transcripts by using the same
draws $(I_i,Y_i)$: the transcripts differ only on the event
$\bigcup_{i\le m}\{B_i=1\}$, whose probability is at most $m\varepsilon$ by the
union bound. Hence
\[
  \tv\big(\mathcal{L}(T_m(f_g)),\,\mathcal{L}(T_m(f_s))\big)
  \le \min\{1,m\varepsilon\}.
\]
By the Neyman--Pearson/Le Cam two-point bound \citep{tsybakov2009introduction},
any decision rule $\psi$ mapping transcripts to $\{\text{grounded},
\text{shortcut}\}$ satisfies
\[
  \max\Big\{\Prb_{f_g}(\psi\neq\text{grounded}),\,
            \Prb_{f_s}(\psi\neq\text{shortcut})\Big\}
  \;\ge\;\frac{1-\tv\big(\mathcal{L}(T_m(f_g)),\mathcal{L}(T_m(f_s))\big)}{2}
  \;\ge\;\frac{(1-m\varepsilon)_+}{2}.
\]
For any fixed budget $m$, letting $\varepsilon\to0$ drives the bound to $1/2$:
no observational audit beats a coin flip in the strong-coupling limit. Note
that the bound applies to any functional of the transcript---accuracy,
AUC, calibration error, explanation scores---because all are measurable
functions of $T_m$.

\textbf{(b) Identifiability under intervention.} The \Mone{} generator realizes
$\dop(Z_E=1)$ and $\dop(Z_E=0)$ at shared $(S,N,s)$: for pair $k$ it renders
$I_k^{\mathrm{sig}}$ and $I_k^{\mathrm{ctr}}$ differing only in $Z_E$. By
construction $f_s(I_k^{\mathrm{sig}})=f_s(I_k^{\mathrm{ctr}})$ for every pair
(the trend statistic is shared), while
$f_g(I_k^{\mathrm{sig}})-f_g(I_k^{\mathrm{ctr}})=2\delta_0>0$ for every pair. In
the noiseless case a single pair separates the classes. If responses are
observed with independent bounded noise, let
$H_k=\ind\{p(I_k^{\mathrm{sig}})>p(I_k^{\mathrm{ctr}})\}$; under $f_s$ the
$H_k$ are Bernoulli$(1/2)$ (exchangeable noise), while under $f_g$ they are
Bernoulli$(1/2+\delta)$ with $\delta>0$ determined by the noise level and
margin $\delta_0$. The test that declares ``grounded'' when
$\bar H_n\ge 1/2+\delta/2$ has, by Hoeffding's inequality, error at most
$\exp(-n\delta^2/2)$ under either hypothesis. The \Mtwo{} argument is
analogous for the narrower claim used in the main text: when non-trend
evidence is balanced across labels, a momentum-only score has population AUC
$>1/2$ on the aligned split and $<1/2$ on the reverse split
(\cref{prop:swap}), while an $S$-invariant score has AUC $1/2$ on both. The
empirical AUC concentrates at rate $\exp(-cn t^2)$ by a U-statistic Hoeffding
bound. Without the evidence-balancing condition, \Mtwo{} is a shortcut stress
test rather than a standalone proof that a model is ungrounded, because
incidental evidence can covary with trend.
\end{proof}

\begin{remark}
The construction uses a deterministic-plus-noise coupling for clarity, but the
argument only needs the observational coupling between $Z_E$ and any
trend-measurable event to be strong; $\varepsilon$ measures the decoupling mass
that observational data offer for free. The benchmark interventions manufacture
decoupling ($\varepsilon=1/2$ in the paired design) rather than waiting for it.
\end{remark}

\subsection{Proof of \texorpdfstring{\Cref{cor:m0-null}}{the Null-Optimality Corollary}}

\begin{proof}
Fix an image $I_i$. In \Mzero-A, the evaluator creates two equally weighted
evaluation items from the same image: $(I_i,1)$ and $(I_i,0)$. The conditional
probability under the benchmark experiment is therefore
$\Prb_{\Mzero\text{-A}}(Y=1\mid I_i)=1/2$. For any strictly proper binary
scoring rule, the conditional risk
\[
  R_i(p)=\tfrac{1}{2}\ell(p,1)+\tfrac{1}{2}\ell(p,0)
\]
is uniquely minimized by reporting the true conditional probability $p=1/2$.
For concreteness: under log loss $R_i'(p)=-\frac{1}{2p}+\frac{1}{2(1-p)}$ with
$R_i''>0$, and under the Brier score $R_i'(p)=2p-1$; both give the unique
minimizer $p=1/2$.

Let $A(I_i)$ be any measurable chart feature, such as past momentum,
volatility, or renderer style. Because the paired construction assigns equal up
and down labels to every individual image, conditioning cannot change the event
probability: $\Prb_{\Mzero\text{-A}}(Y=1\mid A(I_i))=1/2$. Therefore any
systematic variation of $p_i$ with $A(I_i)$ is a property of the model
response, not evidence of benchmark alpha.
\end{proof}

\subsection{Null Extensions for \texorpdfstring{\Mzero-B and \Mzero-C}{M0-B and M0-C}}

\begin{proposition}[Zero-drift stochastic-volatility null]
\label{prop:m0b}
Suppose the shadow future follows
$r_{t+k}=\sigma_{t+k}\epsilon_{t+k}$ with
joint sign symmetry of the future innovations:
\[
(\epsilon_{t+1},\ldots,\epsilon_{t+H})
\mid \mathcal F_t,(\sigma_{t+1},\ldots,\sigma_{t+H})
\overset{d}{=}
-(\epsilon_{t+1},\ldots,\epsilon_{t+H})
\mid \mathcal F_t,(\sigma_{t+1},\ldots,\sigma_{t+H}).
\]
Then $\Prb(R_{t,H}>0\mid\mathcal{F}_t)=1/2$ for any horizon return
$R_{t,H}=\sum_{k=1}^H r_{t+k}$ whenever the conditional distribution of
$R_{t,H}$ is continuous. The condition is satisfied, for example, by
conditionally independent continuous innovations that are each symmetric
around zero.
\end{proposition}

\begin{proof}
Conditional on the volatility path, joint sign symmetry implies
$R_{t,H}\overset{d}{=}-R_{t,H}$. Integrating over the volatility path preserves
symmetry. Continuity rules out mass at zero, giving equal probability to
positive and negative returns.
\end{proof}

\begin{proposition}[Within-bucket balancing]
\label{prop:m0c}
Let $B_i$ be a momentum bucket. If the benchmark samples labels so that
$\Prb(Y=1\mid B_i=b)=1/2$ for every bucket $b$, then a predictor that is a
function only of the bucket has no directional information under a strictly
proper score and is optimized by $p=1/2$ within every bucket.
\end{proposition}

\begin{proof}
The conditional risk given $B_i=b$ has event probability $1/2$. Strict
properness again gives the unique optimum $p=1/2$ for any predictor measurable
with respect to $B_i$.
\end{proof}

\subsection{Proof of \texorpdfstring{\Cref{thm:brier}}{the Calibration--Decision Link}}

\begin{proof}
On \Mzero-A each image is scored against both labels with equal weight, so the
realized Brier risk of reporting $p$ is
$\tfrac12(p-1)^2+\tfrac12 p^2=\tfrac14+(p-\tfrac12)^2$, and the excess over the
optimal report $1/2$ is exactly $\E[(p-\tfrac12)^2]$.

\textbf{OverConf bound.} By Jensen's inequality applied to $x\mapsto x^2$,
$\E[(p-\tfrac12)^2]\ge\big(\E|p-\tfrac12|\big)^2=\mathrm{OverConf}^2$.

\textbf{TBI bound.} Let $Q_{\mathrm{top}},Q_{\mathrm{bottom}}$ be the extreme
momentum quintiles, each of probability mass $q=0.2$, and write
$a=\E[p\mid Q_{\mathrm{top}}]-\tfrac12$,
$b=\E[p\mid Q_{\mathrm{bottom}}]-\tfrac12$, so $\mathrm{TBI}=|a-b|$. Then
\begin{align*}
\E[(p-\tfrac12)^2]
&\ge q\,\E[(p-\tfrac12)^2\mid Q_{\mathrm{top}}]
  + q\,\E[(p-\tfrac12)^2\mid Q_{\mathrm{bottom}}]\\
&\ge q\,a^2+q\,b^2
\;\ge\; \frac{q}{2}\,(a-b)^2 \;=\;\frac{q}{2}\,\mathrm{TBI}^2 ,
\end{align*}
using within-bucket Jensen for the second inequality and
$a^2+b^2\ge(a-b)^2/2$ for the third. Combining the two bounds gives the
maximum.

\textbf{Log loss.} The excess log-loss risk of reporting $p$ on the paired item
is $\mathrm{KL}\big(\mathrm{Ber}(\tfrac12)\,\|\,\mathrm{Ber}(p)\big)$, and by
Pinsker's inequality
$\mathrm{KL}\ge 2(p-\tfrac12)^2$ pointwise; taking expectations and repeating
the two arguments above yields the stated factor-two versions.
\end{proof}

\subsection{Proof of \texorpdfstring{\Cref{prop:ident}}{the Identification Map}}

\begin{proof}
Work on the link scale $u_i=\sigma^{-1}(p_i)$ under \cref{eq:structural}, with
the default prompt/style $(q_0,s_0)$ normalized to $\gamma_{q_0}=\gamma_{s_0}=0$.

(i) On \Mzero{}, $E_i\equiv0$, so $u_i=\alpha+\beta_S\widetilde S_i+\varepsilon_i$
with $\E\widetilde S=0$ and $\varepsilon\perp\widetilde S$. Taking expectations
gives $\alpha=\E u$; taking covariances gives
$\beta_S=\operatorname{Cov}(u,\widetilde S)/\operatorname{Var}(\widetilde S)$.
The link-scale quintile contrast
\[
  \left|\E[u\mid Q_{\mathrm{top}}]-\E[u\mid Q_{\mathrm{bottom}}]\right|
  =
  |\beta_S|\,
  \left|\E[\widetilde S\mid Q_{\mathrm{top}}]
       -\E[\widetilde S\mid Q_{\mathrm{bottom}}]\right|
\]
therefore identifies $|\beta_S|$ whenever the extreme buckets have distinct
mean momentum. The probability-scale TBI reported in the tables is a bounded
diagnostic proxy. Under sign-separated buckets, independent noise, and no
saturation reversal, it has the sign of $\beta_S$ locally around zero; we do
not require or claim global strict monotonicity on the probability scale. For
OverConf, the valid distribution-free Jensen bound is
\[
  \E|p-\tfrac12|\ge |\E p-\tfrac12|.
\]
No lower bound in terms of the link-scale intercept $\alpha$ is used.

(ii) Pair members share $(\widetilde S_i, N_i, q, s)$ and differ only in $E$:
subtracting the two equations,
$u^{\mathrm{sig}}_i-u^{\mathrm{ctr}}_i=\beta_E\Delta E_i+\Delta\varepsilon_i$
with $\E\Delta\varepsilon=0$, so
$\beta_E=\E[u^{\mathrm{sig}}-u^{\mathrm{ctr}}]/\Delta E$, free of
$(\alpha,\beta_S,\gamma)$. If $\Delta\varepsilon$ is symmetric around zero,
the tie-aware hit
\[
  H=\ind\{u^{\mathrm{sig}}>u^{\mathrm{ctr}}\}
    +\tfrac12\ind\{u^{\mathrm{sig}}=u^{\mathrm{ctr}}\}
\]
satisfies $\E H>1/2$ iff $\beta_E\Delta E>0$ and $\E H=1/2$ at the no-response
null. Calibration invariance follows because replacing $p$ by a strictly
increasing $m(p)$ preserves all pairwise orderings and ties.

(iii) On \Mtwo{} there is no injected evidence ($E\equiv0$). For a
momentum-only score with independent noise, and with non-trend evidence
balanced across labels, \cref{prop:swap} applies to the systematic component:
the aligned--reverse AUC gap has the sign of the momentum effect under
non-degenerate trend--label coupling. We use this as an over-identifying
diagnostic for shortcut dependence, not as an unconditional ``if and only if''
test for $\beta_S=0$ in arbitrary response functions.

(iv) For fixed $i$, varying $q$ over the prompt bank at fixed $(s_0)$ gives
$u_{iq}=c_i+\gamma_q+\varepsilon_{iq}$; the within-window variance across
prompts is $\operatorname{Var}_q(\gamma_q)+\operatorname{Var}(\varepsilon)$,
identifying $\operatorname{Var}(\gamma_q)$ up to the noise floor, which is
estimated from repeat queries at fixed $(q,s)$; similarly for renderer
variation and $\operatorname{Var}(\gamma_s)$ via RFI.
\end{proof}

\subsection{Proof of \texorpdfstring{\Cref{prop:completeness}}{the Split Sufficiency Claim}}

\begin{proof}
We exhibit the confounded pairs on the link scale.

\textbf{\Mzero{} alone.} Take $f_1$ with parameters
$(\alpha,\beta_S,\beta_E)=(0,1,0)$ and $f_2=(0,1,1)$. On \Mzero{} every sample
has $E=0$, so the two responders have identical response distributions and
identical values of every \Mzero{} functional (OverConf, TBI, $\rho$,
selective risk), yet one is grounded and the other is not.

\textbf{\Mone{} paired contrast alone.} Take $f_1=(0,0,1)$ and $f_2=(2,5,1)$. The
paired contrast cancels $\alpha+\beta_S\widetilde S$ exactly, so PSS and the
paired mean contrast are identical, yet $f_2$ is heavily trend-biased and
miscalibrated in the null market.

\textbf{\Mtwo{} alone.} On \Mtwo{} windows there is no injected evidence, but
incidental real evidence $Z^{\mathrm{inc}}_E$ covaries with $S$ (the
observational coupling of \cref{thm:impossibility}). Take
$f_1=(0,1,0)$ (pure shortcut) and $f_2=(0,0,1)$ responding only to
$Z^{\mathrm{inc}}_E$. When
$\operatorname{corr}(Z^{\mathrm{inc}}_E,S)\to1$ on the \Mtwo{} window
distribution, both responders induce the same score ordering, hence the same
aligned/balanced/reverse AUCs and the same gap.

\textbf{Core identification and over-identification.} Given \Mzero{} and \Mone{}:
\Mzero{} recovers $(\alpha,\beta_S)$ by \cref{prop:ident}(i); the \Mone{}
contrast recovers $\beta_E$ by \cref{prop:ident}(ii). Thus \Mzero{}+\Mone{}
identify $(\alpha,\beta_S,\beta_E)$ under the working model and matched-pair
assumptions. \Mtwo{} overidentifies the trend coordinate under a flipped
mechanism and acts as a specification test: a discrepancy between \Mzero{}- and
\Mtwo{}-implied trend behavior indicates response nonlinearity, incidental
evidence coupling, or mechanism sensitivity. This is the sense in which the
splits are jointly diagnostic, not a claim that \Mtwo{} is required for point
identification of the three coordinates.
\end{proof}

\subsection{The Trend-Confounder Swap}

\begin{proposition}[Trend-confounder swap]
\label{prop:swap}
If a score $f(I)$ is a monotone function of past momentum and contains no
additional label-related signal, then its AUC is higher on
$\Mtwo_{\mathrm{align}}$ than on $\Mtwo_{\mathrm{reverse}}$ whenever labels are
trend-aligned in the former and trend-reversed in the latter.
\end{proposition}

\begin{proof}
Let $S$ be a scalar past-trend statistic, and suppose the model score is
$f(I)=h(S)$ with strictly increasing $h$. AUC can be written as a Mann--Whitney
probability:
\[
  \mathrm{AUC}=\Prb(f(I^+)>f(I^-))+\frac{1}{2}\Prb(f(I^+)=f(I^-)),
\]
where $I^+$ and $I^-$ are independently sampled positive and negative examples.
Because $h$ is increasing, this is equivalent to comparing $S^+$ and $S^-$.

In the aligned split, the conditional probability $\Prb(Y=1\mid S=s)$ is
increasing in $s$. Therefore the distribution of $S$ among positives
first-order stochastically dominates the distribution of $S$ among negatives,
with strict dominance whenever the label probability varies with $S$. The
aligned AUC is then greater than $1/2$. In the reverse split,
$\Prb(Y=1\mid S=s)$ is decreasing in $s$, so the stochastic ordering is reversed
and the AUC of the same increasing score is less than $1/2$. Thus
\[
  \mathrm{AUC}(\Mtwo_{\mathrm{align}})
  >
  \mathrm{AUC}(\Mtwo_{\mathrm{reverse}})
\]
under non-degenerate trend-label coupling. The balanced split enforces
$\Prb(Y=1\mid S\in b)=1/2$ inside each momentum bucket, so a score that is only a
function of the bucket has no within-bucket ranking information. This proves
that the aligned--reverse gap is a diagnostic for momentum shortcuts rather than
a proof of financial skill. If the score contains another label-related
component, the statement need not hold without balancing that component across
the two split mechanisms.
\end{proof}

\subsection{Proof of \texorpdfstring{\Cref{prop:ces}}{the CES Causal Contrast}}

\begin{assumption}[Edit comparability and locality]
\label{asm:edits}
(i)~\textbf{Matched magnitude:} for each pair, the rendered displacement
satisfies $d(I,I^{E})=d(I,I^{C})$ for a fixed visual metric $d$ along the same
edit path (the released generator matches glyph-level $\ell_2$ displacement).
(ii)~\textbf{Locality:} $T_E$ alters only glyphs rendering $Z_E$ and $T_C$ only
glyphs rendering $Z_C$; the rest of the canvas is bit-identical.
\end{assumption}

\begin{proof}
Under \cref{asm:edits}(ii), the edit pairs realize
$I^{E}=r(T_EZ_E,Z_C,Z_R)$ and $I^{C}=r(Z_E,T_CZ_C,Z_R)$ with all other canvas
content bit-identical, so
\[
  \mathrm{CES}_i=|g(Z_E,Z_C,Z_R)-g(T_EZ_E,Z_C,Z_R)|
               -|g(Z_E,Z_C,Z_R)-g(Z_E,T_CZ_C,Z_R)| ,
\]
and taking expectations gives $\E[\mathrm{CES}]=\kappa$ directly; positivity of
$\kappa$ is the definition of operational faithfulness of the reported
evidence under the edit family. For the bias bound, suppose
\cref{asm:edits}(i) fails. Define the magnitude-matched counterfactual edit
$\tilde I_i^E$ by moving along the same edit path as $I_i^E$ until
$d(I_i,\tilde I_i^E)=d(I_i,I_i^C)$, and set
$\Delta_i^{\mathrm{path}}=d(I_i^E,\tilde I_i^E)$. By $L$-Lipschitz continuity
of $f$ in $d$,
\[
\big||f(I_i)-f(I_i^{E})|-|f(I_i)-f(\tilde I_i^{E})|\big|
\le L\,\Delta_i^{\mathrm{path}} .
\]
The CES computed with $\tilde I_i^E$ has expectation $\kappa$ by the first
part, so $|\E[\mathrm{CES}]-\kappa|\le L\,\E\Delta^{\mathrm{path}}$.
\end{proof}

The diagnostic is intentionally weaker than full causal identification of the
model's internal features. It tests whether the evidence the model itself
reports is more influential than matched controls under the benchmark's edit
family; a non-positive CES does not prove the model ignores all financial
visual features.

\subsection{Invariance Diagnostics}

\begin{proposition}[Invariance diagnostics]
\label{prop:invariance}
Any responder that is a measurable function of $(S,Z_E,N)$ alone---in
particular any grounded responder---is $G$-invariant. RFI and PFI are the
orbit-variance functionals of $G_r$ and $G_q$ respectively, and vanish exactly
on $G$-invariant responders. For the order-two color subgroup $c\in G_r$
(red-up $\leftrightarrow$ green-up), the paired statistic $\E[p\circ c-p]$
tests color-heuristic dependence directly.
\end{proposition}

\begin{proof}
Elements of $G$ act on $(s,q)$ and fix $(S,Z_E,N)$ by definition of
semantic preservation. If $f(I,q)=\phi(S,Z_E,N)$ for measurable $\phi$, then
for any $(g_r,g_q)\in G$,
$f(r(W,g_rs),g_qq)=\phi(S,Z_E,N)=f(r(W,s),q)$: $f$ is constant on orbits.
Grounded responders are of this form by definition. RFI and PFI average the
within-orbit variance of $p$ over base windows; a responder is $G$-invariant
iff every orbit variance is zero (up to query noise), iff both fragility
indices are zero. For the color subgroup $\{e,c\}$ with $c^2=e$, the
orbit is the pair $\{I,cI\}$ and the paired mean $\E[p\circ c-p]$ is zero for
any $c$-invariant responder; a one-sample location test on the paired
differences (we use Wilcoxon signed-rank) therefore tests color-heuristic
dependence directly.
\end{proof}

\subsection{Proof of \texorpdfstring{\Cref{prop:deconfound}}{Overlap-Weighted Artifact Robustness}}

\begin{proof}
Index injected pairs by their background features $X$ (the window summary
statistics available to the generator classifier) and write the pair contrast
$D=\zeta(X)+a(X)+\Delta\varepsilon$, where $\zeta(X)$ is the semantic response
to the injected edit and $a(X)$ is a recorded-background artifact component of
the contrast. The observed \Mone{} pair population has background law
$\Prb_{\Mone}(dx)\propto e(x)\,\Prb(dx)$ while the organic population has
$\Prb_{\Mzero}(dx)\propto(1-e(x))\,\Prb(dx)$, where $\Prb$ is the pooled law.
Weighting \Mone{} pairs by $w(X)=1-e(X)$ produces expectations under the tilted
law $\Prb_{\Mone}^{w}(dx)\propto e(x)(1-e(x))\Prb(dx)$---the overlap
population \citep{li2018balancing,crump2009dealing}---which assigns mass only
where organic windows with the same recorded backgrounds exist. The weighted
estimand is therefore
\[
  \tau_{\mathrm{ow}}
  =
  \E_{\Prb_{\Mone}^{w}}\!\left[D\right]
  =
  \E_{P_{\mathrm{ow}}}\!\left[\zeta(X)+a(X)\right].
\]
Thus overlap weighting changes the target population; it does not algebraically
remove $a(X)$ from an \Mone{} pair contrast. The equality
$\tau_{\mathrm{ow}}=\E_{P_{\mathrm{ow}}}\zeta(X)$ requires an additional
restriction such as pair-invariant artifacts, $\E_{P_{\mathrm{ow}}}a(X)=0$, or
an organic-control contrast that estimates and subtracts $a(X)$. Estimation
replaces $e$ by a grouped cross-validated estimate $\hat e$; consistency for
the overlap-population total contrast follows from consistency of $\hat e$ and
boundedness of $w$, and pair-level bootstrap is used for inference because
pairs are the independent design unit.
\end{proof}

\begin{remark}
The audit reports both unadjusted and adjusted estimates. Stability after
overlap weighting weakens explanations based solely on recorded background
imbalance, but it does not rule out unrecorded artifacts or pair-differential
artifacts without the additional restrictions stated above.
\end{remark}

\subsection{Proof of \texorpdfstring{\Cref{prop:mde}}{the MDE Bounds}}
\label{app:mde}

\begin{proof}
\textbf{PSS.} The tie-aware estimator $\widehat{\mathrm{PSS}}$ is a mean of $n$
i.i.d.\ scores $H_i\in\{0,\tfrac12,1\}$ with null mean $1/2$ and variance at
most $1/4$. The two-sided level-$\alpha$ normal test rejects when
$|\widehat{\mathrm{PSS}}-\tfrac12|>z_{1-\alpha/2}/(2\sqrt n)$, and power
$1-\beta$ against alternative $\tfrac12+\delta$ requires
$\delta\ge(z_{1-\alpha/2}+z_{1-\beta})/(2\sqrt n)$. With $\alpha=0.05$,
$\beta=0.2$: $\delta\ge 2.80/(2\sqrt n)=1.40/\sqrt n$, which is $0.075$ at
$n=351$ scored pairs.

\textbf{TBI.} The quintile contrast is a difference of two means with $qn$
samples each and response standard deviation $\sigma_p$; the standard
two-sample formula gives
$\mathrm{MDE}=(z_{1-\alpha/2}+z_{1-\beta})\,\sigma_p\sqrt{2/(qn)}$. With
$n=800$, $q=0.2$, and the median observed $\sigma_p=0.14$:
$2.80\times0.14\times\sqrt{2/160}=0.044$.

\textbf{CES.} The estimator is a paired mean with empirical pair standard
deviation $\sigma_{\mathrm{CES}}$ (observed $0.05$--$0.12$ across models,
median ${\approx}0.07$), so
$\mathrm{MDE}=(z_{1-\alpha/2}+z_{1-\beta})\,\sigma_{\mathrm{CES}}/\sqrt n$:
$0.011$ at the API budget $n=300$ and $0.004$ at the open-model budget
$n=2000$.

\textbf{Converse.} Distinguishing $\beta_E=0$ from $|\beta_E|=\delta$ from $n$
matched pairs is at least as hard as distinguishing
$\mathrm{Bernoulli}(\tfrac12)^{\otimes n}$ from
$\mathrm{Bernoulli}(\tfrac12+c'\delta)^{\otimes n}$ for the induced hit
sequence (data-processing). The KL divergence between the product laws is
$n\,\mathrm{KL}(\mathrm{Ber}(\tfrac12)\|\mathrm{Ber}(\tfrac12+c'\delta))
\le 4nc'^2\delta^2/(1-4c'^2\delta^2)$, and by the Bretagnolle--Huber
inequality any test with both errors $\le\alpha$ requires the KL to exceed
$\log\big(1/(4\alpha(1-\alpha))\big)$, giving $n\ge c\log(1/\alpha)/\delta^2$.
\end{proof}

\section{Anytime-Valid Sequential Evaluation}
\label{app:evalues}

API evaluation accrues cost per query, and runs are routinely stopped early
(budget limits, provider deprecation, interim results). Classical fixed-$n$
tests are invalid at data-dependent stopping times; the following protocol is
valid when applied to conditionally mean-zero pair scores or to independent
episode/block scores.

\textbf{E-process for pairwise sensitivity.} Order scored pairs
$k=1,2,\dots$ and let
$H_k=\ind\{\text{correct sign}\}+\frac12\ind\{\text{tie}\}\in\{0,\tfrac12,1\}$.
Under the shortcut null $\E[H_k\mid\mathcal F_{k-1}]=1/2$, the betting process
\[
  E_0=1,\qquad
  E_n=\prod_{k\le n}\big(1+\lambda_k\,(H_k-\tfrac12)\big),
  \qquad \lambda_k\in(-2,2)\ \text{predictable},
\]
is a nonnegative supermartingale with $\E E_n\le1$, so by Ville's inequality
$\Prb(\exists n:E_n\ge1/\alpha)\le\alpha$
\citep{ville1939etude,ramdas2023game}. Rejecting the null the first time
$E_n\ge1/\alpha$ is valid at every stopping rule, including ``the
budget ran out'' and ``the run looked done.'' We set $\lambda_k$ by the
aGRAPA/empirical-Kelly rule of \citet{waudbysmith2024betting}, clipped to
$[-1,1]$, which adapts to the unknown effect size; mixture (one-sided) variants
follow \citet{grunwald2024safe}.

\textbf{Dependence.} Pairs within an asset--date episode are dependent.
Randomizing pair order avoids replay artifacts but does not, by itself, make
the conditional null hold under arbitrary within-episode dependence. Validity
under block dependence is obtained by betting on independent episode-level
scores, such as the episode mean tie-aware hit, or by constructing a valid
block-level e-value under an exchangeable/permutation null and multiplying
these e-values across episodes. The fixed-sample results in the paper use the
block bootstrap of \cref{app:asymptotics}; the sequential replay is reported as
episode-level evidence rather than as a pair-level martingale under dependent
pairs.

\textbf{Two-sided and negative effects.} Because the empirical phenomenon
includes anti-semantic responses (conditional sign accuracy below
$1/2$), the release also tracks the mirrored process with
$\lambda_k\mapsto-\lambda_k$; the pair $(E_n^{+},E_n^{-})$ gives an
anytime-valid two-sided test at level $2\alpha$ when the betting units satisfy
the conditional null. The released code reports both pair-order replays
(descriptive, useful for budgeting) and episode-level e-values (valid under the
block-dependence model). This separation keeps the economic value of sequential
monitoring without relying on random order to remove within-episode dependence.

\section{Influence Functions and Block Asymptotics}
\label{app:asymptotics}

All reported metrics are plug-in functionals of paired contrasts. Write
$\theta(P)$ for any of: the paired mean contrast ($\beta_E$ estimand), PSS,
TBI, OverConf, or weighted AUC. Each has a first-order expansion
$\theta(\widehat P_n)-\theta(P)=\frac1n\sum_i\psi_\theta(D_i)+o_p(n^{-1/2})$
with bounded influence function $\psi_\theta$: for means and paired means
$\psi$ is the centered observation; for PSS it is the centered tie-aware hit;
for TBI it is the within-quintile centered response with signs by bucket; for
AUC it is the standard two-sample Hájek projection of the Mann--Whitney
kernel; for the overlap-weighted versions, weights are bounded by construction
($w=1-\hat e\in[0,1]$) so the weighted influence functions remain bounded, with
an additional estimation term for $\hat e$ that is $o_p(1)$ under grouped
cross-fitting.

Under the episode-block dependence model (independence across asset--date
episodes, arbitrary dependence within), the blocked CLT gives
$\sqrt{B}\,(\theta(\widehat P)-\theta(P))\Rightarrow
\mathcal N\big(0,\operatorname{Var}(\bar\psi_b)\big)$, where $B$ is the number
of episodes and $\bar\psi_b$ the within-episode influence average. The released
block bootstrap resamples episodes with all prompt/renderer variants of a base
window kept inside their episode \citep{kunsch1989block}, which consistently
estimates $\operatorname{Var}(\bar\psi_b)$; FDR control across the model
$\times$ metric grid follows \citet{benjamini1995fdr}. This places the released
uncertainty protocol on an estimation-theoretic footing: the bootstrap is the
variance estimator implied by the influence-function expansion, not a
heuristic.

\section{Diagnostic Estimators and Statistical Protocol}
\label{app:diagnostics}

Let $\mathcal{I}$ denote a set of evaluated images and let $p_i$ be the parsed
probability. We use the following estimators.
\[
  \widehat{\mathrm{OverConf}}_{\Mzero}
  =\frac{1}{|\mathcal{I}_{\Mzero}|}\sum_{i\in\mathcal{I}_{\Mzero}}
  |p_i-1/2|.
\]
For top and bottom momentum quintiles $Q_{\mathrm{top}}$ and
$Q_{\mathrm{bottom}}$,
\[
  \widehat{\mathrm{TBI}} =
  \left|
  \frac{1}{|Q_{\mathrm{top}}|}\sum_{i\in Q_{\mathrm{top}}}p_i -
  \frac{1}{|Q_{\mathrm{bottom}}|}\sum_{i\in Q_{\mathrm{bottom}}}p_i
  \right|.
\]
The reported $\rho(p,M)$ is Spearman correlation between $p_i$ and past
momentum $M_i$.

For matched \Mone{} pairs $(I_i^{\mathrm{sig}},I_i^{\mathrm{ctr}})$ with
injected direction $y_i\in\{0,1\}$, pairwise signal sensitivity is the
tie-aware directional hit rate
\[
  \widehat{\mathrm{PSS}} =
  \frac{1}{n}\sum_{i=1}^n
  \left[
  \ind\big\{(2y_i-1)\big(p(I_i^{\mathrm{sig}})-p(I_i^{\mathrm{ctr}})\big)>0\big\}
  +\frac12
  \ind\big\{p(I_i^{\mathrm{sig}})=p(I_i^{\mathrm{ctr}})\big\}
  \right],
\]
whose shortcut/no-directional-information null is $1/2$. We also report the
strict hit rate (ties counted as misses), the tie rate, and conditional sign
accuracy given a non-tie because these quantities separate non-response from
misreading (\cref{sec:results}). M1 AUC is computed on the binary injected-label
task. We report both because AUC can be affected by global score calibration,
whereas PSS is invariant to monotone recalibration
(\cref{prop:ident}(ii)).

Prompt fragility and renderer fragility are paired variance estimates:
\[
  \widehat{\mathrm{PFI}}_i=\mathrm{Var}_{q\in\mathcal{Q}}(p_{iq}),\qquad
  \widehat{\mathrm{RFI}}_i=\mathrm{Var}_{s\in\mathcal{S}}(p_{is}).
\]
For publication tables, these are averaged over base windows; by
\cref{prop:invariance} they are the orbit-variance functionals of the prompt
and renderer groups.

Selective risk uses coverage threshold $\tau$ on either confidence or distance
from $1/2$. Let $A_\tau$ be the accepted set. We report
\[
  \mathrm{coverage}(\tau)=\frac{|A_\tau|}{n},\qquad
  \mathrm{risk}(\tau)=\frac{1}{|A_\tau|}\sum_{i\in A_\tau}
  \ind\{\hat y_i\ne y_i\},
\]
and the full risk--coverage curves in \cref{fig:reliability}. For \Mzero-A,
the preferred behavior is low coverage or probabilities near $1/2$; for
high-signal \Mone{} examples, the preferred behavior is higher coverage with
lower risk.

\subsection{Confidence Intervals and Multiple Comparisons}

Financial windows are temporally and cross-sectionally dependent. We therefore
avoid treating every rendered image as an independent unit. The default
bootstrap unit is an asset--date block or market episode; all prompt and
renderer variants derived from the same base window remain in the same bootstrap
block. For paired comparisons, such as \Mone{} counterfactual pairs and renderer
variants, we resample pairs rather than individual images
(\cref{app:asymptotics} gives the asymptotic justification).

For each metric, we compute percentile bootstrap confidence intervals from
block resamples. When comparing many models and horizons, we control the
false-discovery rate using Benjamini--Hochberg correction. The benchmark release
stores the block assignment for each sample so that confidence intervals can be
recomputed by independent users.

\subsection{Structural Regression and Artifact-Adjustment Implementation}
\label{app:reg-impl}

The structural regression of \cref{tab:joint-regression} estimates
\cref{eq:structural} by OLS on the link scale over all parsed \Mzero/\Mone{}
responses of each model (all styles and prompt variants present in the logs),
with probabilities clipped to $[0.01,0.99]$, momentum standardized, injected
direction $E\in\{-1,0,+1\}$ from the generator records, and style/prompt
dummies. Standard errors are clustered at the base-window level. The
artifact-adjusted estimates of \cref{tab:artifact-adjusted} use a logistic
propensity for $\Prb(\Mone\mid X)$ on the generator-audit features plus
source/interval/regime indicators, fit with 5-fold cross-validation grouped by
symbol hash (AUC $0.892$, consistent with the released audit's $0.933$ from a
boosted model); pair weights are $w=1-\hat e$ evaluated on the neutral pair
member, and intervals are pair-level percentile bootstrap with $2{,}000$
resamples. All recomputations run from released logs without model queries.

\section{Additional Protocol Details}
\label{app:protocol}

\subsection{Data Processing}

All input charts hide ticker symbols, company names, dates, and absolute prices.
Prices are normalized to the first close in the window:
\[
  \widetilde P_{\tau}=100P_{\tau}/C_{t-L+1}.
\]
Volume is transformed using a within-window z-score or min--max normalization,
depending on renderer mode. The main setting uses $L=60$ and $H=5$; robustness
settings use $L\in\{30,60,120\}$ and $H\in\{1,5,10,20\}$. Windows with missing
OHLC values or non-positive prices are removed before rendering.

\subsection{Rendering}

The renderer randomizes color convention, axes, grid, background, moving
averages, volume panels, and aspect ratio. Each rendered image is accompanied by
metadata for candle body boxes, wick coordinates, volume-bar boxes, moving
average intersections, and rule-specific regions such as range boundaries or
the last-five-candle region. These metadata enable causal edits and evidence
evaluation but are not exposed to the evaluated model.

\subsection{Prompt and Parsing Protocol}

The prompt requires a compact JSON object with fields for direction,
probability, confidence, abstention, semantic tags, evidence regions, and a
brief explanation. Outputs that cannot be parsed under this schema are retained
in raw logs but are not treated as successful structured predictions.

We do not request chain-of-thought. The benchmark evaluates observable decisions
and evidence claims, not hidden reasoning. This also reduces the chance that a
model receives credit for a long but unfalsifiable explanation.

\subsection{Release Structure}

The intended release has three layers. The first layer is a source manifest:
data source, download timestamp, symbol list, frequency, date range, field
definition, cleaning rule, and hash; for commercial APIs the manifest also
records the configured model identifier string and evaluation window used for
each run. The second layer is generation code: \Mzero--\Mthree{} samplers,
renderer code, prompt bank, random seeds, and configuration files. The third
layer is the derived benchmark package: images, JSON metadata, labels, parser
outputs, evaluation scripts, and the recomputation scripts for
\cref{tab:joint-regression,tab:artifact-adjusted,fig:reliability}. When a raw
source cannot be redistributed, the release provides scripts and hashes rather
than the raw OHLCV file.

\section{Additional Results}
\label{app:additional-results}

\begin{table}[t]
\centering
\small
\setlength{\tabcolsep}{5pt}
\caption{Main diagnostics. The metrics are paired with a failure mode so that evaluation does not collapse to a single accuracy number.}
\label{tab:metrics}
\resizebox{\linewidth}{!}{%
\begin{tabular}{>{\raggedright\arraybackslash}p{0.20\linewidth}>{\raggedright\arraybackslash}p{0.28\linewidth}>{\raggedright\arraybackslash}p{0.42\linewidth}}
\toprule
Metric & Split & Interpretation \\
\midrule
OverConf, TBI & \Mzero & Confidence or momentum dependence when the benchmark provides no directional information. \\
M1 AUC, PSS & \Mone & Ability to recover injected candlestick evidence and distinguish matched counterfactuals. \\
M2 Gap & \Mtwo & Sensitivity to whether past trend is aligned with, independent of, or opposed to the label. \\
PFI, RFI & Paired prompt/renderer subsets & Fragility under semantically equivalent prompts and visually equivalent renderings. \\
CES & Evidence-editing subset & Whether the model's stated evidence regions are more influential than matched controls. \\
Risk--coverage & All splits with abstention & Whether abstention concentrates errors in uncertain or null conditions. \\
\bottomrule
\end{tabular}
}
\vspace{-0.8em}
\end{table}

\begin{table}[t]
\centering
\small
\setlength{\tabcolsep}{4pt}
\caption{Shadow-market splits and their intended diagnostic roles. The benchmark is designed so that each split has an interpretable ideal response rather than only a leaderboard target.}
\label{tab:splits}
\resizebox{\linewidth}{!}{%
\begin{tabular}{>{\raggedright\arraybackslash}p{0.10\linewidth}>{\raggedright\arraybackslash}p{0.25\linewidth}>{\raggedright\arraybackslash}p{0.29\linewidth}>{\raggedright\arraybackslash}p{0.25\linewidth}}
\toprule
Split & Controlled mechanism & Primary diagnostic & Desired model behavior \\
\midrule
\Mzero & Paired null labels, zero-drift futures, or within-momentum balancing & No-signal overconfidence and trend bias & $\probup\approx 0.5$ or abstain; low dependence on past momentum \\
\Mone & Local OHLCV edits inject rule-defined candlestick evidence with matched counterfactuals & Visual grounding of candlestick semantics & Move probability toward injected signal; cite edited components \\
\Mtwo & Trend-label relation is aligned, balanced, or reversed & Momentum shortcut reliance & Avoid large aligned--reverse performance gap \\
\Mthree & Same local pattern appears under trend, range, or high-volatility regimes & Regime-conditioned confidence and abstention & Adjust coverage and confidence to regime uncertainty \\
\bottomrule
\end{tabular}
}
\vspace{-0.8em}
\end{table}

\begin{table}[htbp]
\centering
\footnotesize
\caption{Traditional non-VLM baselines used for context. LightGBM is included only as a leakage/upper-control diagnostic, not as a model contribution.}
\label{tab:rules}
\begin{tabular*}{\linewidth}{@{\extracolsep{\fill}}llrrrr@{}}
\toprule
Split & Baseline & Acc. & AUC & OverConf & TBI \\
\midrule
M0 & random & -- & -- & 0.000 & 0.000 \\
M0 & past 20 momentum & -- & -- & 0.211 & 0.717 \\
M1 & random & 0.272 & 0.500 & 0.000 & 0.000 \\
M1 & trend & 0.381 & 0.463 & 0.244 & 0.743 \\
M1 & candlestick pattern rule set & 0.417 & 0.964 & 0.113 & 0.095 \\
M1 & lightgbm technical indicators & 0.924 & 0.994 & 0.135 & 0.026 \\
M2 & random & 0.498 & 0.500 & 0.000 & 0.000 \\
M2 & past 20 momentum & 0.502 & 0.504 & 0.211 & 0.715 \\
M3 & candlestick pattern rule set & 0.628 & 0.664 & 0.075 & 0.150 \\
M3 & lightgbm technical indicators & 0.794 & 0.814 & 0.074 & 0.046 \\
\bottomrule
\end{tabular*}
\end{table}

\begin{table}[htbp]
\centering
\footnotesize
\caption{Diagnostic positive controls. These trained controls are not part of the benchmark-only contribution; they verify that the benchmark can reward models explicitly optimized for injected shadow-market rules.}
\label{tab:dcc-controls}
\begin{tabular*}{\linewidth}{@{\extracolsep{\fill}}lrrrrrr@{}}
\toprule
Model & M0 OverConf & M0 TBI & M1 AUC & PSS & M2 Gap & CES \\
\midrule
DCC-VLM-Small & 0.099 & 0.057 & 0.662 & 0.675 & 0.111 & -0.009 \\
DCC-VLM-Full & 0.157 & 0.114 & 0.963 & 0.909 & -0.465 & 0.207 \\
DCC-VLM-Full-Cal & 0.058 & 0.012 & 0.952 & 0.832 & -0.019 & 0.097 \\
\bottomrule
\end{tabular*}
\end{table}

\begin{table}[t]
\centering
\footnotesize
\caption{Benchmark leakage and artifact audits. M0 metadata AUC near 0.5 supports the null-label construction; generator distinguishability is disclosed as a diagnostic risk for some splits.}
\label{tab:audit}
\begin{tabular*}{\linewidth}{@{\extracolsep{\fill}}lrp{0.6\linewidth}@{}}
\toprule
Audit & Value & Interpretation \\
\midrule
M0 metadata AUC & 0.499 & M0 metadata leakage check; ideal is near 0.5. \\
Generator clf. M0 & 0.933 & High values indicate split/source artifacts to disclose as diagnostic risk. \\
Generator clf. M1 & 0.933 & High values indicate split/source artifacts to disclose as diagnostic risk. \\
Generator clf. M2 & 0.498 & High values indicate split/source artifacts to disclose as diagnostic risk. \\
Generator clf. M3 & 0.581 & High values indicate split/source artifacts to disclose as diagnostic risk. \\
\bottomrule
\end{tabular*}
\end{table}

\end{document}